\documentclass[a4paper,11pt]{article}
\pdfoutput=1 

\usepackage[skip=10pt plus1pt, indent=5pt]{parskip}
\usepackage{jheppub}
\usepackage[T1]{fontenc}
\usepackage[english]{babel}
\usepackage{graphicx}
\usepackage{float}
\usepackage[dvipsnames]{xcolor}

\usepackage{amsthm}

\hypersetup{
      breaklinks=true,
      linkcolor=BlueGreen,
      urlcolor=NavyBlue,
      citecolor=Purple,
      filecolor=black,
      menucolor=black,
      pdfstartview=FitV,
      bookmarksopen=true
      }

\newtheorem{thm}{Theorem}
\newtheorem{lemma}[thm]{Lemma}
\newtheorem{cor}[thm]{Corollary}

\newcommand\dd{\mathrm{d}}
\newcommand\Sext{\Sigma_{+}}
\newcommand\Sint{\Sigma_{-}}
\newcommand\Lext{L_{+}}
\newcommand\Lint{L_{-}}
\newcommand\vpar{\vartheta}

\title{Toroidal black holes in four dimensions}

\author[a]{Carlos Barceló,}
\author[b]{Gerardo Garc\'ia-Moreno,}
\affiliation[a,b]{Instituto de Astrof\'{\i}sica de Andaluc\'{\i}a (IAA-CSIC), Glorieta de la Astronom\'{\i}a, 18008 Granada, Spain}

\author[c]{Alejandro Jiménez Cano}
\affiliation[c]{Escuela T\'ecnica Superior de Ingenier\'ia de Montes, Forestal y del Medio Natural, Universidad Politécnica de Madrid, 28040 Madrid, Spain}

\emailAdd{carlos@iaa.es}
\emailAdd{gerargar.moreno@gmail.com}
\emailAdd{alejandro.jimenez.cano@upm.es}

\abstract{
From a purely geometric (kinematic) perspective, black holes in four dimensional spacetimes can have event horizons with arbitrary topologies. It is only when energy conditions are imposed that the horizon's topology is constrained to be that of a sphere. Despite this, exploring exotic horizon topologies remains theoretically intriguing since it allows to unveil structural aspects of General Relativity and gain intuition on energy condition violations. In the axisymmetric case, besides the well-known spherical topology, only a toroidal topology is consistent with the symmetry. Complete solutions, describing the entire exterior region of such toroidal black holes without singularities, have not been reported yet. To the best of our knowledge, the construction we present here is the first explicit example of a toroidal black hole solution in four spacetime dimensions that is free of singularities in the external region.
}

\begin{document}

\maketitle

\section{Introduction}
\label{Sec:Introduction}

Black holes are among the most fascinating predictions of General Relativity. Kinematically (on purely geometrical grounds) it is possible to imagine black holes with horizons of various shapes and topologies. However, the no-hair theorems, which are dynamical since they rely on Einstein equations, in addition to a bunch of additional assumptions constrain the metric of asymptotically flat general relativistic black holes to the biparametric Kerr family, see~\cite{Chrusciel2012}. A key assumption in deriving these results is that the entire external region, extending to asymptotic infinity, is vacuum. Relaxing this condition to allow for matter in the external region introduces the possibility of deforming the black hole. In the static and axisymmetric case it has been rigorously proven by some of the authors and collaborators~\cite{Barcelo2024} that an external matter distribution completely determines the shape of the horizon, as there is only one equilibrium configuration for the horizon of a black hole immersed in a given matter distribution. Any other shape would lead to an irregular geometry with infinite curvatures at the horizon. This result offers additional perspectives to a previous related analysis~\cite{Gurlebeck2015}. We expect that similar results would hold in more general situations, e.g. stationary instead of static configurations. 

Furthermore, in four spacetime dimensions, having stationary horizons with non-spherical topology requires that the external matter violates energy conditions~\cite[Prop. 9.3.2]{Hawking1973}. While in the case of toroidal black holes, this exotic matter could live far from the horizon, for higher genus this matter must permeate the horizon itself. Although such configurations are unlikely to be relevant in astrophysical contexts, since energy condition violations are expected only in extreme causal scenarios or near singularities, not in macroscopic settings, it is still interesting to analyze them from a theoretical point of view. Studying these exotic horizons provides valuable insights into the structural aspects of general relativity and helps develop intuition about the implications of energy condition violations.

In the restricted framework of staticity and axisymmetry, requiring that the axial Killing field does not degenerate at the horizon limits the possible horizon topologies to only spherical and toroidal. While the properties of local vacuum black holes in this setup (those with an open neighborhood around the horizon free of matter) were previously characterized~\cite{Geroch1982} (see also~\cite{Xanthopoulos1983}), explicit solutions extending to an asymptotically flat infinity were not explicitly provided. Invoking a classic theorem by Hawking~\cite[Prop. 9.3.2]{Hawking1973}, one can conclude that achieving an asymptotically flat region would either require violations of energy conditions or result in a singularity in the external region. Following earlier work by Thorne on static and axisymmetric configurations~\cite{Thorne1975}, Peters~\cite{Peters1979} attempted to construct a black hole with toroidal topology. However, the attempt revealed that configurations with regular event horizons inevitably led to naked singularities in the exterior region. The main goal of this paper is to present an alternative construction of toroidal black holes that avoid these singularities in the external region.

Black holes with toroidal horizons in four spacetime dimensions have previously been reported in the literature, but only in the presence of a negative cosmological constant (see~\cite{Senovilla2014} and references therein). From the perspective of Hawking’s classic result forbidding nontrivial horizon topologies, the cosmological constant effectively contributes to the energy-momentum tensor and leads to violations of the energy conditions throughout the spacetime, including at the horizon. In this respect, such solutions are not locally vacuum, unlike ours, which crucially rely on the fact that toroidal horizons can be compatible with a locally vacuum geometry.

Furthermore, Ref.~\cite{Kleihaus2019} presents toroidal black hole solutions in five spacetime dimensions sourced by a phantom scalar field, a matter content that violates the energy conditions, which apparently extends all the way to the horizon. Furthermore, the singularities exhibit a conical singularity, similar to the Emparan-Reall static black rings~\cite{Emparan2001}. The authors conjecture that analogous configurations might exist in four dimensions, although the field equations are not explicitly solved in that case. Moreover, the proposed line element for such a four-dimensional scenario does not clearly describe a \emph{vacuum} toroidal horizon (meaning that there is a neighborhood of the horizon that is Ricci flat) and may therefore differ from the configurations we analyze. Additional discussion of this line element can be found in~\cite[Sec. IV.C]{Cunha2024}.

Here is an outline of the article. In Sec.~\ref{Sec:Literature} we collect and present in an ordered way previous results and analyses available in the literature. We describe in Subsec.~\ref{Subsec:Theorems} previous general results about the local and global properties of toroidal black holes in static and axisymmetric setups and in Subsec.~\ref{Subsec:Peters} we review the attempt by Peters of building a toroidal black hole based on the previous work by Thorne. Secs.~\ref{Sec:Construction}-\ref{Sec:EC} constitute the core of the article. In Sec.~\ref{Sec:Construction} we present the construction of toroidal black holes interpolating from a flat spacetime region to a locally vacuum black hole region. We first show in Subsec.~\ref{Subsec:Try1} why a single matching hypersurface is insufficient to join the two regions; we then build  a model with two shells and an intermediate interpolating region in Subsec.~\ref{Subsec:Try2}; and show in Subsec.~\ref{Subsec:Try_one_shell} how one of the shells can be eliminated by adequately choosing the interpolating region. In Sec.~\ref{Sec:EC} we discuss where the violations of energy conditions appear. We argue that such violations are due to two features that occur simultaneously through the geometry and illustrate them separately with two toy models in Subsecs.~\ref{SubSec:InfinteTorus}-\ref{Subsec:ExternalThinShell}. We then prove in Subsec.~\ref{Subsec:Finite} that qualitatively these violations are generic features of our model near the shell matching with the external region. We then provide some examples with specific interpolating regions and the specific profile of energy condition violations in Subsec.~\ref{Subsec:specific}. We conclude in Sec.~\ref{Sec:Conclusions} with a summary of the work and a discussion of future interesting lines of work. 

 \paragraph*{\textbf{Notation and conventions.}}

We use the signature $(-,+,+,+)$  for the spacetime metric and we work in geometrized units in which $c=G_N=1$. For the curvature tensors we use the conventions in the book of Wald~\cite{Wald1984}, i.e., $[\nabla_a, \nabla_b] V^c =: - R_{abd}{}^c V^d$, $R_{ab}:=R_{acb}{}^c$. We use lower, case indices from the beginning of the alphabet as indices in the 4-dimensional manifold $(a,b...)$ and lower case indices from the middle of the alphabet $(i,j...)$ for indices within the 3-dimensional submanifolds defined in the text. When taking limits, we will also denote as $\lim_{x \searrow a}$ the limit from the right, and $ \lim_{x \nearrow a}$ the limit from the left. When dealing with functions defined on an interval, e.g. $f(x)$ with $x \in [a,b]$, the limit of its derivative approaching the boundaries of the interval will be denoted as $f'(a) := \lim_{x \searrow a} f'(x)$, $f'(b) := \lim_{x \nearrow b} f'(x)$ and similarly for higher-order derivatives, as well as for partial derivatives. We use the symbol $\sim$ to indicate that only the dominant terms of an asymptotic expansion have been kept. The symbol $\Delta$ used at several places in the work denotes a difference and not the Laplacian, which is always denoted by $\nabla^2$.

\section{Overview of previous literature}
\label{Sec:Literature}

\subsection{General results}
\label{Subsec:Theorems}

All the static and axisymmetric toroidal black holes that are locally in vacuum\footnote{
    By \emph{locally in vacuum} we mean that the energy-momentum tensor vanishes in a neighborhood of the horizon.} 
were characterized in the seminal work of Geroch and Hartle~\cite{Geroch1982}, see also~\cite{Xanthopoulos1983} for additional insights. To be more precise, any static and axisymmetric metric can be locally expressed in Weyl coordinates as:
\begin{align}
    \dd s^2= - e^{2U} \dd t^2 + e^{-2U} \left[  e^{2V} \left( \dd r^2 +   \dd z^2 \right) +  r^2 \dd \varphi^2 \right]\,,
    \label{Eq:Weyl}
\end{align}
for certain functions $U=U(r,z)$ and $V=V(r,z)$. The angular coordinate is periodically identified $\varphi \sim \varphi + 2 \pi$, and it is normalized so that the vector $\partial_{\varphi}$ generating the axial symmetry has orbits of period $2 \pi$. The vacuum Einstein equations take the following form for this ansatz:
\begin{align}
    & \partial_r^2 U + \frac{1}{r} \partial_r U + \partial_z^2 U = 0, \label{Eq:Laplace} \\
    & \partial_r V = r \left[ (\partial_r U)^2 - ( \partial_z U)^2 \right], \label{Eq:DrV} \\
    & \partial_z V = 2 r \partial_r U \partial_z U. \label{Eq:DzV}
\end{align}
Notice that Eq.~\eqref{Eq:Laplace} is an axisymmetric Laplace equation for $U$ in the three-dimensional fiduciary Euclidean space described in cylindrical coordinates $(r,z,\varphi)$, as $U$ does not depend on $\varphi$, while Eqs.~\eqref{Eq:DrV}-\eqref{Eq:DzV} are two first-order equations. They can be straightforwardly integrated to obtain $V$ once $U$ has already been determined. The event horizon of a static configuration corresponds to a surface of infinite redshift~\cite{Carter1973}. From the point of view of the Laplace equation, such a solution exhibits some distributional source, although at $r = 0$, which strictly speaking is not covered by this coordinate system, see~\cite{Barcelo2024} and references therein. For the line element in Eq.~\eqref{Eq:Weyl}, an infinite redshift surface corresponds to $U \to - \infty$, so the event horizon must be located precisely at $r \to 0$ in order for its spatial sections to have finite area (since this requires $g_{\varphi\varphi}$ to remain bounded). Therefore, these coordinates are valid for $r>0$, and they break down as $r \to 0$.

The simplest toroidal black hole that one can think of can be constructed following a simple procedure. Take a solution of the Laplace equation corresponding to having as a source an infinite rod located at $r=0$. In fact we are solving a Poisson equation with a distributional source, but as we will see, this source is actually ``fictitious'', meaning it is not associated with any real energy-momentum tensor. The functions $U,V$ associated with this source are
\begin{align}
    U_T &= \log \left( \frac{r}{2m} \right), \label{Eq:UT} \\
    V_T &= \log \left( \frac{r}{m} \right). \label{Eq:VT}
\end{align}
As the solution does not depend on $z$, we can construct a compactified version of it by just identifying $(t, r, z=m,\varphi) \equiv (t, r, z = -m, \varphi + \alpha)$.  This $\alpha$ is a twisting angle and can be safely taken to be zero.
The line element becomes:
\begin{equation}
    \dd s^2= - \left(\frac{r}{2m}\right)^2 \dd t^2 + 4\left[ \dd r^2 +   \dd z^2  + m^2 \dd \varphi^2 \right].
    \label{Eq:GerochHartleBH}
\end{equation}
At this point, we recall the Rindler line element~\cite[p. 149]{Wald1984} 
\begin{equation}
    \dd s^2= - a^2 {\rm x}^2~\dd t^2 + \dd {\rm x}^2 + \dd {\rm y}^2  + \dd {\rm z}^2,
    \label{Eq:Rindler}
\end{equation}
which describes the Minkowski spacetime in coordinates adapted to an observer undergoing constant proper acceleration $a$ in the Cartesian ${\rm x}$ direction. We can recognize that the previous line element~\eqref{Eq:GerochHartleBH} is simply the Rindler metric with the transverse coordinates $({\rm y},{\rm z})$ periodically identified. Therefore, the spacetime is no longer a patch of Minkowski space (with $\mathbb{R}^4$ topology) but a flat toroidal semi infinite tube with $r$ its longitudinal direction (i.e., the spatial sections of the exterior metric have topology $\mathbb{T}^2 \times \mathbb{R}^+$).\footnote{We use the notation $\mathbb{T}^2$ for the 2-torus $\mathbb{S}^1 \times \mathbb{S}^1$.} In fact, through a change of coordinates~\cite{Barcelo2024}:
\begin{align}
    & R^2 - T^2 = 4 r^2, 
    & \frac{T}{R} = \tanh \left(\frac{t}{4m}\right),
\end{align}
we can make the metric to be explicitly flat
\begin{align}
    \dd s^2 = - \dd T^2 + \dd R^2 + 4 \dd z^2 + 4m^2 \dd  \varphi^2. 
\end{align}
Thus, it is clear that there is no matter content associated with this solution. Furthermore, the new coordinates allow the solution to be extended to the interior of the black hole region. The situation is completely parallel to the static and axisymmetric but topologically spherical vacuum black hole, i.e., the Schwarzschild black hole. There, the solution corresponds to a ``fictitious'' source for the Laplace equation which is a constant density rod of finite length $2m$.

Alternatively, one could begin by considering Weyl metrics where the functions $U$ and $V$ are periodic in $z$, thereby imposing a global toroidal topology from the outset. The simple geometry described above corresponds to the solution of the Laplacian in a flat toroidal fiducial space, associated with a ``fictitious'' source modeled as a uniform closed rod density. Since the Laplace equation is linear, any additional solution $\Delta U$ can be added to an existing one to obtain a new solution $\tilde{U}= U + \Delta U$.\footnote{
    We recall that the function $\tilde{V}$ in the metric can be obtained by integrating Eqs.~\eqref{Eq:DrV}-\eqref{Eq:DzV} for the function $\tilde{U}$.}
The resulting geometry displays a regular horizon, provided that $\Delta U$ is analytic at $r = 0$. Physically, the added term $\Delta U$ corresponds to extra sources located away from the original, ``fictitious'' source at $r = 0$; specifically, they are supported wherever the Laplacian of $\Delta U$ is nonzero, i.e., where $\nabla^2 (\Delta U) \neq 0$. These modifications can be interpreted as distortions of the toroidal black hole\footnote{
    In terms of the embedding of constant $t$ surfaces in Euclidean space this corresponds geometrically to a finite toroidal tube in the same way that the Schwarzshcild metric corresponds to a finite spherical tube.}
induced by external matter sources~\cite{Geroch1982,Barcelo2024}. Indeed, such deformations are required in order to have a toroidal horizon if we want to ensure asymptotic flatness.

Now, we want to find not just toroidal black hole tubes, but asymptotically flat toroidal black holes. It is when trying to connect the toroidal geometry at the horizon with an asymptotically flat geometry, that one finds that matter external to the horizon must violate at least the Dominant Energy Condition (DEC)~\cite{Hawking1971}. The purpose of this paper is precisely to build specific models of static and axisymmetric toroidal black hole spacetimes that are locally in vacuum (with no cosmological constant either) and such that they are asymptotically flat. Our construction in fact violates the Null Energy Condition (NEC), and as such all the rest of the standard pointwise energy conditions, namely the Weak Energy Condition (WEC), the DEC and the Strong Energy Condition (SEC).

\subsection{Locally vacuum toroidal black holes: previous attempts}
\label{Subsec:Peters}

An attempt to construct an explicit toroidal black hole was done by Peters~\cite{Peters1979} following up on a previous static and axisymmetric solution to the Einstein vacuum equations discussed by Thorne~\cite{Thorne1975}. However, Peters found out that a regular event horizon for that specific solution requires the presence of some singularities in the external region. Let us discuss it in detail. 

We start by describing Thorne's solution. The idea is to consider the Einstein vacuum equations \eqref{Eq:Laplace}-\eqref{Eq:DzV} for a static and axisymmetric metric expressed in Weyl coordinates~\eqref{Eq:Weyl}. In this case, the potential $U$ is taken to be originated in a line segment singularity (we will later see that this is generically promoted to a real spacetime naked singularity) lying between $z = -a$ and $z = + a$ in the $z$-axis $(r=0)$. In addition, the potential is constrained to have $\partial_z U=0$ at two disks located at $z= \pm a$ and having radius $b$. Thorne shows that far away from the constraining disks, i.e., for $\sqrt{r^2 + z^2 } \gg b $, the potential is almost spherical with $U = - M/r + \mathcal{O}(r^{-2}) $ and $V = \mathcal{O}(r^{-2})$, and hence the metric is asymptotically flat with mass $M$. 

From the point of view of the solution to the Laplace equation, it corresponds to the solution which obeys the following conditions: 
\begin{enumerate}
    \item The solution behaves as $U  \sim (M/a) \log (r) $ as $r \rightarrow 0$ with $z \in (-a,a)$.
    \item $\partial_z U \to 0$ as $z \rightarrow \pm a$ with $r \in (0,b)$. 
    \item $\nabla^2 U = 0$ everywhere for $r >0$. 
    \item $U \sim - M / \sqrt{r^2 + z^2}$ as $\sqrt{r^2+z^2} \rightarrow \infty$.
    \item $V$ obeys Eqs~\eqref{Eq:DrV}-\eqref{Eq:DzV} everywhere and vanishes as $\sqrt{r^2+z^2} \rightarrow \infty$. 
\end{enumerate}
To make the topology of the source toroidal, one needs to include additional ingredients which are the identifications of the outer faces of the upper and lower disks, i.e., 
\begin{align}
    \lim_{\epsilon \searrow 0} (t,r,z = a + \epsilon , \varphi) = \lim_{\epsilon \searrow 0} (t,r,z = - a - \epsilon , \varphi), \quad r \in (0,b],
\end{align}
and the inner faces of the upper and lower disks too
\begin{align}
    \lim_{\epsilon \searrow 0} (t,r,z = a - \epsilon , \varphi) = \lim_{\epsilon \searrow 0} (t,r,z = - a + \epsilon , \varphi), \quad r \in (0,b].
\end{align}
This automatically ensures that the singularity has a toroidal topology. The solution contains three free parameters $(M,a,b)$. In general, it is not possible to find a closed analytic expression for the profile of $U$ (nor $V$ of course), but several approximations are available for $a/b \ll 1$ and $b/a \ll1$ respectively, see~\cite{Thorne1975,Peters1979} for the expressions. 

At this point, once the metric is characterized, one can realize that there are two potential spacetime singularities associated with this metric: the one at $r \rightarrow 0$ for $z \in (-a,a)$ and another one which is precisely lying at the edge of the disks (the common edge, given the periodic identification), i.e., $r \to b$ and $z = \pm a$. Thorne (see Sec. IIF of~\cite{Thorne1975}) identified that removing the latter requires imposing a constraint among the three parameters $(M,a,b)$. Regarding the other putative singularity at $r \rightarrow 0$ for $z \in (-a,a)$, we have that the asymptotic form of the metric as $r \rightarrow 0$ is:
\begin{align}
    \dd s^2  = -  r^{2M/a} \dd t^2 + r^{(2M/a) \times (M/a - 1)} (\dd r^2 + \dd z^2) + r^{-2(M/a-1)} \dd  \varphi^2 + (\text{finite terms}).
\end{align}
The condition $M/a = 1$ is required to have a smooth horizon at $r=0$. Indeed, for $M/a = 1$, Peters~\cite{Peters1979} identified that this leading order for the metric is precisely the flat spacetime metric in Rindler coordinates, where the coordinate $z$ has been periodically identified, i.e., the line element that we discussed above for $U =U_T$ and $V = V_T$, see Eqs.~\eqref{Eq:UT}-\eqref{Eq:VT}. Moreover, he also showed that the subleading terms of the solution still provide regular curvatures as $r \to 0$, hence allowing for the extension of the solution to the inside region (they are analytic $\Delta U$ and $\Delta V$ pieces of the ones discussed above).  However, we also saw that those pieces need to arise from some form of matter content or singular content in the outside region. In fact, in this case they arise from the singularity at the edges of the disks. Peters showed (see Sec. II~\cite{Peters1979}) that removing the singularity at the edges of the disks requires that $M/a >2$, which is clearly incompatible with the $M/a = 1$ requirement of a smooth horizon. A pictorial representation of the solution for arbitrary values of the parameters can be found in Fig.~\ref{Fig:Thorne_Solution}.
\begin{figure}
\begin{center}
\includegraphics[scale=1.1]{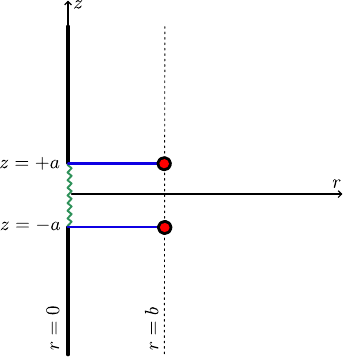}
\caption{Pictorial representation of Thorne toroidal solution. The singularity at $r = 0$ and $z \in (-a,a)$ is depicted as a zigzag green line, and the disks that are identified are depicted in blue. The red dots at the edge of the disks represent the putative singularity that arises.}
\label{Fig:Thorne_Solution}
\end{center}
\end{figure} 
Thus, it is impossible within this solution to simultaneously eliminate the singularity at $r \to 0$ and the external singularities at the edges of the disks. One could have anticipated this based on Hawking's theorem since there is no matter content in the outside region.

\section{Building toroidal black holes}
\label{Sec:Construction}
In this section, we provide an intuitive construction of toroidal black holes without singularities in the external region. The geometries presented here describe the transition from a flat spacetime to a region containing a toroidal black hole horizon, the latter being locally described by the line element in Eq.~\eqref{Eq:Weyl} and the functions in Eqs.~\eqref{Eq:DrV}-\eqref{Eq:DzV}. For that purpose, we distinguish three regions of spacetime that we call $\Omega_1$, $\Omega_2$ and $\Omega_3$. The regions $\Omega_2$ and $\Omega_3$ are delimited by the values of a certain coordinate $\ell$, $\ell\in(\Lint, \Lext)$ and $\ell\in (0, \Lint)$, respectively, whereas $\Omega_1$ represents an outermost flat spacetime region with a torus removed. Notice that, by construction, we are taking $0<\Lint<\Lext$. In addition, there will be two special matching hypersurfaces, $\Sext$ and $\Sint$, at $\ell= \Lext$ and $\ell=\Lint$ joining, respectively, $\Omega_1$ and $\Omega_2$, and $\Omega_2$ and $\Omega_3$. See Fig.~\ref{Fig:Om1Om2Om3} for a pictorial representation.

\begin{figure}
    \centering
    \includegraphics[scale=1.1]{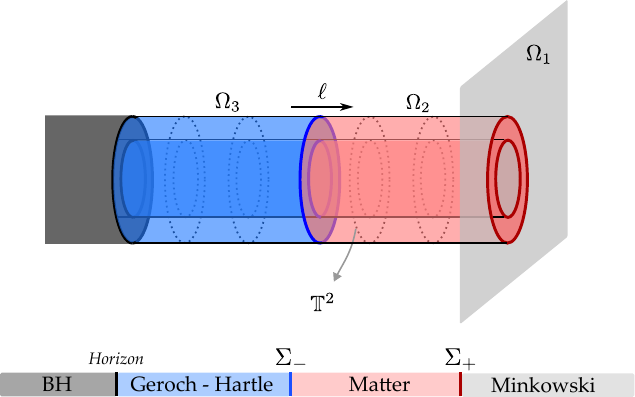}
    \caption{Schematic representation of the toroidal black hole constructed in this work.}
    \label{Fig:Om1Om2Om3}
\end{figure}

\paragraph*{\textbf{Exterior flat region $\Omega_1$}.}
The first region $\Omega_1$ is a flat spacetime region described by the set of coordinates $\{X^a\}=\{T,X,Y,Z\}$ in terms of which the metric looks like
\begin{align}
    \dd s^2_1 = -\dd T^2 + \dd X^2 + \dd Y^2 + \dd Z^2\,.
    \label{Eq:Region1}
\end{align}
This represents the asymptotic external region of the black hole and extends to a surface $\Sext$ with the topology of a torus, which in this set of coordinates is described by the following equation
\begin{align}
    \left( \sqrt{X^2 + Y^2} - \Lext \right)^2 + Z^2 =  b^2,
    \label{Eq:TorusEq}
\end{align}
where $\Lext$ is the radial distance from the origin to the central circumference of the tube and the positive parameter $b$ is the radius of the tube. See Fig.~\ref{Fig:Toro}. Note that the coordinates $(X,Y,Z)$ are restricted to values satisfying $( \sqrt{X^2 + Y^2} - \Lext )^2 + Z^2 >  b^2$. The metric induced in the surface $\Sext$ is clearly that of a curved torus. Due to Gauss--Bonnet theorem we know that the total integrated scalar curvature must be zero. Thus, this surface contains regions with positive as well as regions with negative intrinsic scalar curvature.
\begin{figure}
\begin{center}
\includegraphics[scale=1.1]{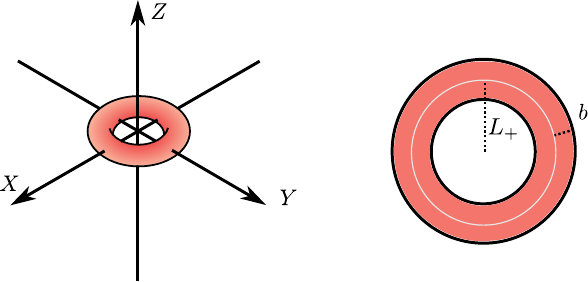}
\caption{Representation of the outermost region $\Omega_1$ for a fixed time $T$ and the limiting hypersurface $\Sext$ (in red).}
\label{Fig:Toro}
\end{center}
\end{figure} 

%
\paragraph*{\textbf{Interior flat region $\Omega_3$}}
We also consider a region $\Omega_3$ which is described by the following line element: 
\begin{align}
\label{Eq:gOm3}
    \dd s^2_3 = - \frac{\ell^2}{m^2} \dd t^2 + \dd \ell^2 +  \frac{m^2}{\pi^2}  \dd z^2 +  m^2 \dd \varphi^2,
\end{align}
where $z \sim z +2 \pi$ and $\varphi \sim \varphi + 2 \pi$. This metric describes the same geometry as \eqref{Eq:GerochHartleBH}. Indeed the expression \eqref{Eq:gOm3} can be obtained from Eq.~\eqref{Eq:GerochHartleBH} by performing the redefinitions:
\begin{equation}
    z\to \frac{m}{\pi}z\,, \qquad t\to 2t\,,\qquad r \to \frac{\ell}{2}\,,
\end{equation}
followed by $m\to m/2$. Here, the coordinate $\ell$ takes values in the open interval between $0$ and a certain finite value $\ell=\Lint$. The latter corresponds to the hypersurface $\Sint$, which for any fixed value of the time coordinate has the geometry of a flat torus (vanishing 2-dimensional intrinsic curvature). Therefore, it corresponds to a finite length tube foliated by flat toroidal 2-surfaces.
\paragraph*{\textbf{Interpolating region $\Omega_2$}}
Finally, we consider the region $\Omega_2$ that will interpolate between the flat torus $\Sint$ and the curved torus $\Sext$, respecting the Darmois--Israel junction conditions \cite{Darmois1927, Israel1966}. For this region, we can take a line element of the form: 
\begin{equation}\label{eq:region2General}
    \dd s^2_2 = - H(\ell) \dd \tau^2 + \dd \ell^2 + \mathcal{F}(\ell, \beta) \dd \alpha^2 + b^2 \dd  \beta^2\,,
\end{equation}
with the angular coordinates periodically identified $\alpha \sim \alpha + 2 \pi$ and $\beta \sim \beta + 2 \pi$, and with $H$ and $\mathcal{F}$ positive functions satisfying the boundary conditions:
\begin{equation}\label{eq:calFH_interpcond}
    H(\Lint)=H(\Lext)=1\,,\quad \mathcal{F}(\Lint, \beta) = \frac{m^2}{\pi^2}\,,\quad \mathcal{F} (\Lext, \beta) =  \left(\Lext + b \cos \left(\beta\right) \right)^2.
\end{equation}
A specifically simple choice of interpolating function is 
\begin{equation}\label{eq:FG}
    \mathcal{F}(\ell, \beta) =  F(\ell) \left( \Lext+ b \cos (\beta) \right)^2 + G(\ell)\,,
\end{equation}
which makes the line element to be of the form:
\begin{equation}\label{eq:gregion2}
    \dd s^2_2 = - H(\ell) \dd \tau^2 + \dd \ell^2 + \left[ F(\ell) \left( \Lext+ b \cos (\beta) \right)^2 + G(\ell) \right] \dd \alpha^2 + b^2 \dd  \beta^2.
\end{equation}
Notice that this region is spatially foliated by 2-dimensional tori, for each value of $\tau$ and $\ell$ (dotted lines in Fig.~\ref{Fig:Om1Om2Om3}). In particular, $\beta$ runs around the loop of the torus and $\alpha$ is the angle of revolution of the torus, which corresponds to an isometric direction (i.e., $\partial_\alpha$ is a Killing vector of \eqref{eq:gregion2}). The coordinates $\alpha$ and $\beta$ are depicted in Fig.~\ref{fig:alphabeta} for clarity. Moreover, the interpolation conditions for this choice read:
\begin{equation}\label{eq:FGH_interpcond}
    H(\Lint)=H(\Lext)=1\,,\quad F(\Lint) = 0\,,\quad F(\Lext) = 1\,,\quad G(\Lint) = \frac{m^2}{\pi^2}\,,\quad G(\Lext) = 0\,.
\end{equation}

\begin{figure}
\begin{center}
\includegraphics[scale=1.1]{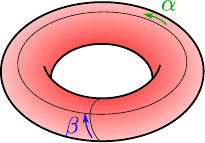}
\caption{Geometrical meaning of the coordinates $\alpha$ and $\beta$.}
\label{fig:alphabeta}
\end{center}
\end{figure} 

This region is confined between $\ell = \Lext$ which describes the surface $\Sext$ and $\ell=\Lint$ corresponding to $\Sint$. In the remainder of the article we will specifically focus on this line element for the sake of concreteness, although we will prove a general result for the line element~\eqref{eq:region2General} in Sec.~\ref{Sec:EC} since the proof is straightforwardly generalizable to that case. The simplest example that we can think of is the following choice of functions:
\begin{align}\label{eq:exampleFGH}
    F(\ell) = \dfrac{\ell^2-\Lint^2}{\Lext^2-\Lint^2}\,,\qquad G(\ell) = \dfrac{m^2}{\pi^2} \dfrac{\Lext^2-\ell^2}{\Lext^2-\Lint^2}\,,\qquad H(\ell)=1\,.
\end{align}

\subsection{No model with a single thin-shell}
\label{Subsec:Try1}

At first sight, one might wonder why the interpolating region $\Omega_2$ is necessary, rather than simply using a thin shell to match the exterior and interior. The main problem that one faces is that there would be a discontinuity in the metric, given that the induced metric on that shell would necessarily inherit a curved metric from the outside and a flat one from the inside. 

\begin{figure}
\begin{center}
\includegraphics[scale=1.2]{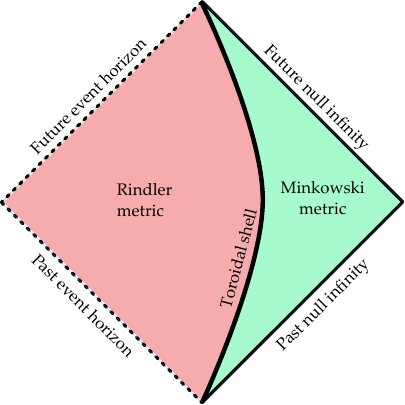}
\caption{Pictorial representation of the simplest setup to match an exterior flat spacetime with an interior Rindler metric with a periodic identification of two of the transverse coordinates. The matching would need to occur on a surface with the topology of a torus. However, given that the intrinsic metrics that the torus inherits from the inside and the outside are different (whereas the one of the inside is flat, the one of the outside is curved), it is not possible to match them without a singular behavior.}
\label{Fig:Toroidal_Shell_Try}
\end{center}
\end{figure} 

Explicitly, imagine that we want to match an external empty flat spacetime region with an interior region described by the line element~\eqref{Eq:Weyl}. This means that we take the spacetime to be of the form~\eqref{Eq:Weyl} for $r \in (0,R)$, and we match it with a Minkowski spacetime from which we have removed the interior of a torus at every time. See Fig.~\ref{Fig:Toroidal_Shell_Try} for a pictorial representation. Although both surfaces display the same topology, it is not possible to match them smoothly within General Relativity. The first of the Darmois--Israel junction conditions in General Relativity requires that the metric is continuous across the matching surface in order to have a well-defined distributional curvature. The continuity of the metric implies that the two induced metrics (one from each side) need to correspond to the same intrinsic geometry. In the matching scenario of our toy model the torus inherits from the inside region a flat geometry, i.e., it has no intrinsic curvature, whereas the induced metric obtained from the exterior region is curved.

One might wonder, whether it could be possible to choose a different smooth embedding $\Sext$ of a torus in $\mathbb{R}^3$  allowing for a flat induced metric from $\Omega_1$. However, it is not possible to construct a $C^{2}$-embedding of this kind. To see this, let us proceed by \emph{reductio ad absurdum}. Suppose there exists such an embedding of the torus in flat space with identically vanishing Gauss curvature. Given that it is a compact surface, it would be possible to enclose it with a sphere, which has positive Gaussian curvature at every point. Now, suppose that we reduce the radius of the sphere until it touches the torus tangentially for the first time. On the one hand, we have kept the torus inside the sphere during all this process. On the other hand, the Gaussian curvatures of both surfaces disagree at the contact point and, since they are tangent, the surface of higher curvature (the sphere) must lie inside the other one, contradicting the previous statement. This proves the impossibility to embed a flat torus in flat space with at least $C^2$-differentiability. This forbids this simple construction and requires us to perform some more subtle matching between the external and internal region by allowing for an interpolation four-dimensional region where there is a non-distributional matter content. 
 
It is still interesting to note that if one relaxes the $C^2$ requirement and only demands $C^1$-differentiability, this proof breaks down (since the Gauss curvature is no longer defined for the torus). In fact, it was shown by the seminal theorems of Nash and Kuiper~\cite{Nash1954,Kuiper1955,Kuiper1955b} that $C^1$-embeddings of flat torus into the flat spacetime actually exist. To make sense of this physically, one would need to interpret a discontinuity on the metric by some physical argument.\footnote{Notice however that $C^2$-embeddings of the flat torus are possible in more than 3 spatial dimensions, e.g. the Clifford torus in $\mathbb{R}^4$.}

\subsection{Model with an interpolating region and two thin shells}
\label{Subsec:Try2}

The regions $\Omega_1$ and $\Omega_3$ are flat and hence correspond to vacuum solutions. The region $\Omega_2$ displays a nontrivial matter content. In addition, generic simple interpolations do also have distributional matter at the matching surfaces $\Sint$ and $\Sext$. In fact, the distributional energy-momentum tensor $S_{ij}$ at those hypersurfaces can be directly obtained from the Darmois--Israel junction conditions as
\begin{align} \label{eq:SK}
    8  \pi S_{ij} = - [[K_{ij}]] + [[K]] h_{ij} \,,
\end{align}
where $h_{ij}$ is the appropriate induced metric and the double square brackets of a quantity, $\left[ \left[ \mathcal{O} \right] \right]$, represents its jump, (i.e., the size of the discontinuity) across the surface. Normal directions will always be taken pointing outwards (from $\Omega_3$ to $\Omega_2$ and from $\Omega_2$ to $\Omega_1$).

\paragraph*{\textbf{Energy-momentum tensor in region $\Omega_2$}}
In this section we compute the energy-momentum tensor supporting the non-flat geometry in the region $\Omega_2$. 
The only non-vanishing components of the energy-momentum tensor are the following:
\begin{align}
    8 \pi T^{\tau}{}_{\tau}& = -\frac{1}{4b^2}\left[\left(\frac{\partial_\beta \mathcal{F}}{\mathcal{F}}\right)^2 + b^2 \left(\frac{\partial_\ell \mathcal{F}}{\mathcal{F}}\right)^2 - 2 \left(  
    \frac{\partial^2_\beta \mathcal{F}}{\mathcal{F}} + b^2 \frac{\partial^2_\ell \mathcal{F}}{\mathcal{F}} \right) 
    \right]\,, \\[2ex]
    8 \pi T^{\ell}{}_{\ell} & = -\frac{1}{4b^2}\left[\left(\frac{\partial_\beta \mathcal{F}}{\mathcal{F}}\right)^2 - 2   
    \frac{\partial^2_\beta \mathcal{F}}{\mathcal{F}} - b^2 \frac{H'}{H}\frac{\partial_\ell \mathcal{F}}{\mathcal{F}} \right]
    \,,\\[2ex]
    8 \pi T^{\ell}{}_{\beta} & = \frac{1}{4}\left[\frac{\partial_\beta \mathcal{F}\partial_\ell \mathcal{F}}{\mathcal{F}^2} - 2 \frac{\partial_\beta\partial_\ell \mathcal{F}}{\mathcal{F}} \right] \qquad = 8 \pi b^2 T^{\beta}{}_{\ell}\,, \\[2ex]
    8 \pi T^{\alpha}{}_{\alpha} & 
    = \frac{1}{4}\left[2\frac{H'' }{H }-\left(\frac{H'}{H}\right)^2\right]\,,\\[2ex] 
    8 \pi T^{\beta}{}_{\beta} & 
    = \frac{1}{4}\left[ 2
    \frac{\partial^2_\ell \mathcal{F}}{\mathcal{F}}
    -\left(\frac{\partial_\ell \mathcal{F}}{\mathcal{F}}\right)^2
    +2\frac{H''}{H}-\left(\frac{H'}{H}\right)^2
    +\frac{H'}{H}\frac{\partial_\ell \mathcal{F}}{\mathcal{F}}\right]\,,
\end{align}
where we have omitted the coordinate dependencies of the functions $\mathcal{F},H$ and their derivatives.  It has the form
\begin{equation}
    \begin{pmatrix}
        T^{\tau}{}_{\tau} & 0 & 0 & 0 \\
        0 & T^{\ell}{}_{\ell} & 0 & T^{\ell}{}_{\beta} \\
        0 & 0 & T^{\alpha}{}_{\alpha} & 0 \\
        0 & T^{\beta}{}_{\ell} & 0 & T^{\beta}{}_{\beta}
    \end{pmatrix}\,,
\end{equation}
which has a nontrivial timelike eigenvector ($\partial_\tau$) and therefore belongs to the Segré-Pleba\'nski class [111,1] \cite[Ch. 5]{Stephani2003} (type I in Hawking-Ellis' notation \cite[Sec. 4.3]{Hawking1973}).

\paragraph*{\textbf{Distributional energy momentum tensor in $\Sext$}}

Let us go first with the shell that is located in $\Sext$. For that purpose, we take the parametrization of the torus in terms of the coordinates $\{y^i\} = \{\tau, \alpha, \beta\}$:
\begin{align}
    & T ( \tau) = \tau, \\
    & X (\alpha, \beta) = \cos (\alpha) \Phi(\beta), \\
    & Y (\alpha, \beta) = \sin (\alpha) \Phi(\beta), \\
    & Z (\alpha, \beta) = b  \sin (\beta)\,, 
\end{align}
where we have introduced the convenient notational shortcut:
\begin{equation}
    \Phi(\beta) := \Lext + b \cos(\beta)\,.
    \label{Eq:Phidef}
\end{equation}
In terms of these coordinates, we can compute the following basis forms in $\Sext$
\begin{align}
    e^{a}{}_{i} = \frac{\partial X^a}{\partial y^i},  
\end{align}
which leads to
\begin{align}
    e^T{}_i \dd y^i& = \dd \tau, \\
   e^X{}_i \dd y^i&= - \sin (\alpha) \Phi(\beta) \dd \alpha - b \cos (\alpha) \sin (\beta) \dd \beta, \\
    e^Y{}_i \dd y^i& = \cos (\alpha) \Phi(\beta) \dd \alpha - b \sin (\alpha) \sin (\beta) \dd \beta, \\
   e^Z{}_i \dd y^i& = b \cos (\beta) \dd \beta.
\end{align}
The induced metric on the surface reads:
\begin{align}\label{eq:h+Region1}
    (\dd s_+^2)_{\Omega_1} = - \dd \tau^2 +\Phi(\beta)^2 \dd \alpha^2 + b^2 \dd \beta^2\,. 
\end{align}
We can now compute the unit normal vector to the surface and we find:
\begin{align}
    (n_+^a)_{\Omega_1} = \frac{1}{b} \left[ \left(1- \frac{\Lext}{\sqrt{X^2+Y^2}}  \right)X  \delta^{a}{}_{X} + \left(1- \frac{\Lext}{\sqrt{X^2+Y^2}}  \right)Y \delta^{a}{}_{Y} + Z \delta^{a}{}_{Z} \right]
\end{align}
The extrinsic curvature of the surface $\Sext$, as seen from outside region $\Omega_1$ is:
\begin{align}\label{eq:KijOmega1}
    (K^+_{ij})_{\Omega_1} = e^a{}_i e^b{}_j(\nabla_a n_{+b})_{\Omega_1} .
\end{align}
Given that the spacetime is flat and it is expressed in Minkowski coordinates, we have that the Christoffel symbols vanish and we can substitute $\nabla\to \partial$ in \eqref{eq:KijOmega1}. We have the following non-vanishing components:
\begin{align}
    & (K^+_{\alpha \alpha})_{\Omega_1} =  \Phi(\beta) \cos(\beta) , \\
    & (K^+_{\beta \beta})_{\Omega_1} = b . 
\end{align}
The trace reads:
\begin{align}
    (K^+)_{\Omega_1} = \frac{1}{b} + \frac{\cos(\beta)}{\Phi(\beta)} \,.
\end{align}
We can now repeat the exercise from the inside region $\Omega_2$. The surface, which corresponds to $ \ell = \Lext$, is parametrized by $\{\tau, \alpha, \beta\}$
\begin{align}
    (\dd s_+^2)_{\Omega_2} = - \dd \tau^2 + \Phi(\beta)^2 \dd \alpha^2 + b^2 \dd \beta^2.
    \label{Eq:CurvedTorus}
\end{align}
In this case, the vector normal to the surface is
\begin{align}
    (n^a_+)_{\Omega_2} = \delta^a{}_r, 
\end{align}
and we can again compute the non-vanishing components of the extrinsic curvature to find:
\begin{align}
    (K^+_{\tau\tau})_{\Omega_2} &=  -\frac{1}{2} H'(\Lext),\\
    (K^+_{\alpha \alpha})_{\Omega_2} &= \frac{1}{2} \partial_\ell \mathcal{F}(\Lext, \beta_+),
\end{align}
and the trace reads:
\begin{align}
    (K^+)_{\Omega_2} = \frac{1}{2}\left[H'(\Lext) + \frac{\partial_\ell\mathcal{F}(\Lext,\beta)}{\Phi(\beta)^2}\right] \,.
\end{align}
The only nontrivial discontinuities in the extrinsic curvatures at $\Sext$ are given by:
\begin{align}
    [[ K^{+}_{\tau\tau} ]] &= \frac{1}{2} H'(\Lext), \label{Eq:K+tautau}\\
    [[ K^{+}_{\alpha\alpha} ]] &= \Phi(\beta)\cos(\beta) -\frac{1}{2} \partial_\ell \mathcal{F}(\Lext, \beta_+) , \\ 
    [[ K^{+}_{\beta\beta}]] &= b,  \label{Eq:K+bebe}\\
    [[K^{+}]] &= \frac{1}{b}+\frac{\cos(\beta)}{\Phi(\beta)}-\frac{1}{2} H'(\Lext)-\frac{1}{2}\frac{\partial_\ell\mathcal{F}(\Lext,\beta)}{\Phi(\beta)^2},
\end{align}
and hence the non-trivial components of the energy-momentum tensor read: 
\begin{align}
    & 8 \pi S^+_{\tau\tau} = -\frac{1}{b}- \frac{\cos(\beta)}{\Lext + b \cos (\beta)} + \frac{\partial_\ell \mathcal{F}(\Lext,\beta)}{2(\Lext + b \cos (\beta))^2}, \label{Eq:S+tautau} \\
    & 8 \pi S^+_{\alpha\alpha} =  (\Lext + b \cos (\beta))^2\left(\frac{1}{b}-\frac{1}{2}H'(\Lext)\right), \label{Eq:S+alpha alpha}\\
    & 8 \pi S^+_{\beta\beta} = \left[\frac{\cos(\beta)}{\Lext + b \cos (\beta)}-\frac{1}{2}H'(\Lext)-\frac{\partial_\ell \mathcal{F}(\Lext,\beta)}{2(\Lext + b \cos (\beta))^2}\right]b^2\,, \label{Eq:S+betabeta}
\end{align}
where we have inserted $\Phi(\beta)$ back using Eq.~\eqref{Eq:Phidef}.
\paragraph*{\textbf{Distributional energy momentum tensor in $\Sint$}.}

Now we repeat the exercise with the shell located at $\Sint$. From the outside region $\Omega_2$ we take $ \ell = \Lext$ and take the parametrization $\{\tau, \alpha, \beta\}$ to get the following induced metric on the surface:
\begin{align}\label{eq:ds-Om2}
   (\dd s_-^2)_{\Omega_2} = -\dd  \tau^2 + \frac{m^2}{\pi^2} \dd  \alpha^2 + b^2 \dd  \beta^2,
\end{align}
and the extrinsic curvature is
\begin{align}
    (K^-_{\tau\tau})_{\Omega_3} &= -\frac{1}{2} H'(\Lint) \,,\\
    (K^-_{\alpha\alpha})_{\Omega_3} &= \frac{1}{2} \partial_\ell\mathcal{F}(\Lint,\beta)\,. 
\end{align}
The trace in this case reads:
\begin{align}
    (K^-)_{\Omega_2} =  \frac{\pi^2}{2m^2} \partial_\ell\mathcal{F}(\Lint,\beta) +\frac{1}{2} H'(\Lint)\,.
\end{align}
From the inside region $\Omega_3$, $\Sint$ corresponds with the hypersurface with $ \ell = \Lint$. If we parameterize with $\{t = (m/\Lint )\tau, z=\alpha, \varphi=\beta\}$ we reach the following expression for the induced metric:
\begin{align}\label{eq:h-Region3}
    (\dd s_-^2)_{\Omega_3} = -  \dd \tau^2 + \frac{m^2}{\pi^2} \dd \alpha^2 + m^2 \dd  \beta^2\,.
\end{align}
Continuity with \eqref{eq:ds-Om2} requires:
\begin{equation}
    b^2 = m^2\,,
\end{equation}
which can be used to eliminate the parameter $m$ in favor of $b$. The extrinsic curvature from the inside can be determined by realizing that also the normal vector to the surface is
\begin{align}
    (n_-^a)_{\Omega_3} = \delta^a{}_{\ell}\,, 
\end{align}
and we find that there is only one non-vanishing component of the extrinsic curvature:
\begin{align}
    (K^-_{\tau\tau})_{\Omega_3} = - \frac{1}{\Lint}\,,
\end{align}
with the trace being 
\begin{align}
    (K^-)_{\Omega_3} = \frac{1}{\Lint} \,. 
\end{align}

The non-vanishing jumps in the extrinsic curvatures at $\Sint$  are given by:
\begin{align}
    & [[K^{-}_{\tau\tau}]] = -\frac{1}{2} H'(\Lint)+\frac{1}{\Lint}, \\
    & [[ K^{-}_{\alpha\alpha} ]] =  \frac{1}{2}\partial_\ell\mathcal{F}(\Lint,\beta) , \\ 
    & [[K^{-}]] = \frac{\pi^2}{2b^2} \partial_\ell\mathcal{F}(\Lint,\beta) +\frac{1}{2} H'(\Lint) -\frac{1}{\Lint},
\end{align}
leading to the following non-trivial components for the distributional energy-momentum tensor at $\Lint$ (after substituting $\Phi$):
\begin{align}
    8 \pi S^-_{\tau\tau} & = -\frac{\pi^2}{2b^2}\partial_\ell\mathcal{F}(\Lint,\beta) , \\
    8 \pi S^-_{\alpha\alpha} & = \frac{b^2}{\pi^2}\left[\frac{1}{2}H'(\Lint) -\frac{1}{\Lint}\right] , \\
    S^-_{\beta\beta} & =\pi^2 S^-_{\alpha\alpha} - b^2 S^-_{\tau\tau}  . 
\end{align}

\subsection{Model with an interpolating region and only an external thin shell}
\label{Subsec:Try_one_shell}

By appropriately choosing the functions $\mathcal{F}(\ell,\beta)$ and $H(\ell)$, we can match both regions without invoking a matter shell at $\ell=\Lint$. Indeed, we just need:
\begin{equation}
     \partial_\ell\mathcal{F}(\Lint,\beta) = 0\quad \forall\beta\,,\qquad H'(\Lint) = \frac{2}{\Lint}\,,\label{eq:no-shell}
\end{equation}
for the distributional energy-momentum tensor $S^-_{ij}$ to identically vanish. A simple set of functions satisfying the interpolation conditions \eqref{eq:calFH_interpcond} (namely, \eqref{eq:FGH_interpcond}) and the condition of having no shell at $\ell=\Lint$ \eqref{eq:no-shell} is given by an ansatz of the type \eqref{eq:FG} with the choices:
 \begin{align}
     F(\ell) &= \left(\frac{\ell-\Lint}{\Lext-\Lint}\right)^2\,,\nonumber\\
     G(\ell) &= \dfrac{m^2}{\pi^2} \frac{(\Lext-\ell)(\ell+\Lext-2\Lint)}{(\Lext-\Lint)^2}\,,\nonumber\\
     H(\ell)&=\frac{-2\ell^2+(\Lext+\Lint)(2 \ell-\Lint)}{\Lint(\Lext-\Lint)}\,,
\end{align}
all of which are strictly positive in the range $\ell\in(\Lint,\Lext)$.

It is not possible to do the same with the external shell for this class of geometries in which the line element for the interpolating region is given by Eq.~\eqref{eq:region2General}. The absence of shell implies the vanishing of the right hand-side of Eqs.~\eqref{Eq:S+tautau}-\eqref{Eq:S+betabeta}, which is equivalent to the vanishing of Eqs.~\eqref{Eq:K+tautau}-\eqref{Eq:K+bebe} due to \eqref{eq:SK}. The vanishing of Eq.~\eqref{Eq:K+bebe} would require $b=0$, which automatically contradicts our assumption of having a non-trivial external toroidal surface ($b>0$).

\section{Energy conditions violations}
\label{Sec:EC}

Classic theorems ensure that violations of energy-conditions need to occur outside the horizon to have a toroidal topology~\cite[Prop. 9.3.2]{Hawking1973}. On a heuristic basis, we expect that the energy-conditions need to be violated since: i) the flat spatial torus of the Geroch--Hartle line element is converted into a curved one so that it can match the flat exterior, ii) the region between the horizon and infinity interpolates between a tube and an asymptotic flat region with spheres whose area increases without a bound as we approach infinity. Both processes require violations of energy conditions, and for some choices of interpolating functions these violations can occur simultaneously in the same region, making their interpretation more difficult. This is for instance what would happen if we had a model without shells. The next two subsections are dedicated to the independent study of the violations of energy-conditions due to these two processes to gain intuition, before analyzing the full model that we have introduced.  

\subsection{Developing some intuition: infinite toroidal tube}
\label{SubSec:InfinteTorus}

To understand where and how the geometry violates the energy conditions due to the curved metric of the torus, it is instructive to look first at a simpler case. Consider the spacetime $\mathbb{R}^2\times \mathbb{T}^2$ with metric
\begin{align}
    \dd s^2 = - \dd  \tau^2 + \dd \ell^2 + \left[ \Lext + b \cos(\beta) \right]^2  \dd \alpha^2 + b^2 \dd\beta^2 \,,
\end{align}
whose spatial part is an infinite tube of toroidal sections inheriting a curved metric for every constant $(\tau, \ell)$. We are taking $\Lext> b$ again. This metric violates the NEC and consequently all the pointwise energy conditions~\cite{Kontou2020}. To prove it, it is sufficient to take the null vector field $k^a := (\partial_\tau)^a + b^{-1}(\partial_\beta)^a$ and contract it with the Einstein tensor to find
\begin{equation}
    G_{ab} k^a k^b =\frac{1}{b} \frac{\cos(\beta)}{\Lext + b\cos(\beta)}\,,
\end{equation}
which becomes negative for $\beta \in (\pi/2, 3\pi/2)$.

As an intuitive way of understanding why energy conditions are violated, we can consider a congruence of geodesic observers at constant $\alpha=\alpha_0$ and $\ell=\ell_0$, i.e., those timelike curves going around the $\beta$ direction. For instance, we can take the parametrization
\begin{equation}\label{Eq:geodesicszeta}
    (\tau, \ell_0, \alpha_0, \beta(\tau)) \qquad \text{with}\qquad \beta(\tau)= \frac{A}{b} \tau\quad \text{for}\quad A\in(0,1),
\end{equation}
whose corresponding normalized tangent vector is given by:
\begin{equation}
    \zeta^a = \frac{1}{\sqrt{1-A^2}}\left((\partial_\tau)^a + \frac{A}{b}(\partial_\beta)^a\right)\,.
\end{equation}
This congruence is twist-free and the expansion is given by:
\begin{equation}
    \theta(\tau) = - \frac{A}{\sqrt{1-A^2}} \frac{\sin(\beta(\tau))}{\Lext + b \cos(\beta(\tau))}\,.
\end{equation}
The derivative with respect to the parameter is:
\begin{equation}
    \frac{\dd \theta}{\dd \tau}  = - \frac{A^2}{\sqrt{1-A^2}} \frac{b+\Lext\cos(\beta(\tau))}{\Big[\Lext + b \cos(\beta(\tau))\Big]^2}\,,
\end{equation}
whose sign depends on the sign of the function $b+\Lext\cos(\beta(\tau))$ which is changing since $\Lext> b$. This indicates that different observers of this congruence become closer or separate, and hence violations of the SEC occur depending on the value of $\beta$, see Fig.~\ref{fig:flatcurvedtorusdeform}. 
\begin{figure}[H]
    \centering
    \includegraphics[scale=1]{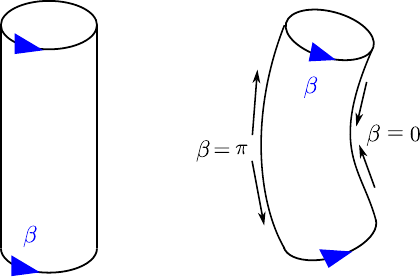}
    \caption{Pictorial representation of the geometrical operation (stretching-compressing) required to go from a flat to a curved torus.}
    \label{fig:flatcurvedtorusdeform}
\end{figure}
%

\subsection{Developing intuition: external thin shell}
\label{Subsec:ExternalThinShell}

Aside from the energy-condition violations that arise due to the presence of curved torus spatial slices, we expect also energy condition violations that arise from the matching between the toroidal tube and an asymptotically flat region. Although one could think of doing it smoothly, the simplest way of doing it is by introducing the distributional energy-momentum tensor of a thin shell. In our model, this corresponds to the external thin-shell. It is easy to understand that it must violate the NEC and, with it, all the other pointwise energy conditions. Null rays traveling parallel to the tube toward the thin shell will enter the flat spacetime where they will diverge, i.e., they will go apart from each other. 

To isolate the violations of energy-conditions arising from this matching, we can work out a simpler example before considering the energy-condition violations occurring in our full model. Consider taking the infinite torus from the previous section, but now cutting it at $\ell = \Lext$, i.e., we take the line element
\begin{align}
    \dd s^2 = - \dd  \tau^2 + \dd \ell^2 + \left[ \Lext + b \cos(\beta) \right]^2  \dd \alpha^2 + b^2 \dd\beta^2 \,,
    \label{Eq:InfiniteTubeCut}
\end{align}
for $\ell \leq \Lext$. We glue it at $\ell = \Lext$ where the induced metric reads
\begin{align}
    \dd s^2 = - \dd \tau^2 +  \left[ \Lext + b \cos(\beta) \right]^2  \dd \alpha^2 + b^2 \dd\beta^2 \,, 
\end{align}
and the external region of the asymptotically flat space from which we have removed a torus, i.e., we take the line element
\begin{align}
    \dd s^2_1 = -\dd T^2 + \dd X^2 + \dd Y^2 + \dd Z^2\,
\end{align}
for the region
\begin{align}
    \left( \sqrt{X^2 + Y^2} - \Lext \right)^2 + Z^2 \geq b^2.
\end{align}
The equality holds for the matching surface. In fact, the Darmois--Israel junction conditions have already been worked out for the matching between these geometries as the line element in Eq.~\eqref{Eq:InfiniteTubeCut} simply corresponds to the one in Eq.~\eqref{eq:region2General} upon the identification $H(\ell) = 1$ and $\mathcal{F}(\ell, \beta) = \left( \Lext + b \cos \left( \beta \right) \right)^2 $. Thus, by plugging this specific choice of functions in Eqs.~\eqref{Eq:S+tautau}-\eqref{Eq:S+betabeta} we find 
\begin{align}
    & 8 \pi S^+_{\tau\tau} = -\frac{1}{b}- \frac{\cos(\beta)}{\Lext + b \cos (\beta)} ,\\
    & 8 \pi S^+_{\alpha\alpha} =  \frac{\left(\Lext + b \cos (\beta)\right)^2}{b}, \\
    & 8 \pi S^+_{\beta\beta} = \frac{b^2 \cos(\beta)}{\Lext + b \cos (\beta)}\, . 
\end{align}
It corresponds to the energy-momentum tensor of a fluid (type I in Hawking-Ellis classification~\cite[Sec. 4.3]{Hawking1973}) on the surface with density
\begin{equation}
    \sigma = - \frac{1}{b} - \frac{\cos (\beta) }{\Lext+ b\cos (\beta)}
\end{equation}
and anisotropic pressures
\begin{align}
    p_{\alpha} & = \frac{1}{b}, \\ 
    p_{\beta} & = \frac{\cos (\beta) }{\Lext+ b\cos (\beta)}.
\end{align}
It is immediate to conclude that this energy-momentum tensor violates the NEC. To see this, recall that the NEC requires $\sigma + p_i \geq 0$ for $i = \alpha,\beta$.\footnote{Note that, for type I energy-momentum tensors according to Hawking-Ellis classification~\cite[Sec. 4.3]{Hawking1973}, this is equivalent to demanding that $T_{ab}l^a l^b\geq 0$ for every null vector $l^a$.} In our case, we find:
\begin{align}
    & \sigma + p_{\alpha} = - \frac{\cos (\beta) }{\Lext+ b\cos (\beta)}, \\
    & \sigma + p_{\beta}  = - \frac{1}{b}.
\end{align}
We trivially have $\sigma + p_{\beta} <0$ everywhere, and  $\sigma + p_{\alpha}<0$ for $\beta \in [0, \pi/2) \cup (3\pi/2,2\pi)$. A pictorial representation of the phenomenon can be found in Fig.~\ref{fig:NECcurves}.
\begin{figure}[H]
    \centering
    \includegraphics[scale=1.4]{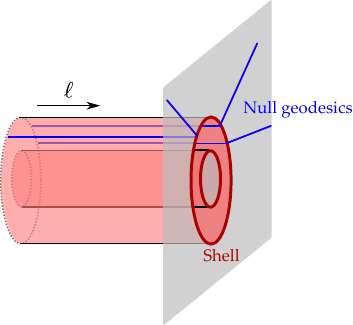}
    \caption{Pictorial representation of the light rays diverging when they cross the external shell and reach the asymptotic flat region.}
    \label{fig:NECcurves}
\end{figure}
%

\subsection{General analysis near the end of the interpolating region}
\label{Subsec:Finite}

The proof for the finite interpolating torus region is simply an adaptation of the observation from the previous subsection. In fact, given that close to $\Lext$ the solution behaves in a very similar way to the infinite torus, energy condition violations are also expected.  

We start by recalling the WEC, which requires that for any future-pointing timelike vector field $\xi^a$
\begin{equation}
    T_{ab} \xi^a \xi^b \geq 0\,. 
\end{equation}
This can be restricted to normalized timelike vector fields meaning that the energy density measured by all of them is everywhere non-negative. Similarly, the NEC says that for any future-pointing null vector field $l^a$
\begin{equation}
    T_{ab} l^a l^b \geq 0\,. 
\end{equation}
One important well-known result is that WEC implies NEC, which can be proved by assuming WEC and approaching  in a continuous way the desired null vector $l^a$ by a sequence of timelike vectors. As a consequence of this, if we find a violation of the NEC in an open set then WEC must also be violated in such a region. 

To study potential violations of these conditions we consider the following one-parameter family of vector fields:
\begin{equation}
    (A_\vpar)^a := H(\ell)^{-1/2} (\partial_\tau)^a + \vpar b^{-1} (\partial_\beta)^a\,,
\end{equation}
with $\vartheta \in [0,1]$ and define the scalar
\begin{equation}\label{eq:scalarEC}
        \mathcal{Z}_\vpar := G_{ab} (A_\vpar)^a (A_\vpar)^b = 8\pi T_{ab}\, (A_\vpar)^a (A_\vpar)^b\,,
\end{equation}
where we made use of the Einstein field equations to relate the Einstein tensor with the energy-momentum tensor.

This family of vector fields interpolates between two situations:
\begin{itemize}
    \item Case $\vpar = 0$. This corresponds to the observer following the coordinate lines of $\partial_\tau$, namely 
    \begin{equation}\label{eq:observertau}
        u^a := (A_0)^a = H(\ell)^{-1/2} (\partial_\tau)^a\,.
    \end{equation}
    Then the scalar \eqref{eq:scalarEC} is proportional to the energy density $T_{ab}u^a u^b$ measured by this observer:
    \begin{equation}
        \mathcal{Z}_0 = 8\pi T_{ab}u^a u^b = \frac{1}{H(\ell)} G_{\tau\tau} \,.
    \end{equation}
    Wherever $\mathcal{Z}<0$, the WEC is violated. We highlight at this point that $\mathcal{Z}_0\geq 0$ does not imply that WEC is satisfied, since $\mathcal{Z}_0$ only represents the energy density measured by the particular observer $u^a$, whereas the WEC requires the energy density for \emph{any} observer to be non-negative.
    
    \item Case $\vpar = 1$. This corresponds to a null vector field
    \begin{equation}
        k^a := (A_1)^a = H(\ell)^{-1/2} (\partial_\tau)^a +  b^{-1} (\partial_\beta)^a\,,
    \end{equation}
    and the scalar $\mathcal{Z}_1$ can be used to explore violations of the NEC. As in the previous case, if $\mathcal{Z}_1<0$ the NEC (and hence the WEC too) is violated. But $\mathcal{Z}_1\geq 0$ gives no information about the NEC since we are focusing on a particular null vector, $k^a$.
\end{itemize} 
In particular, for the metric \eqref{eq:region2General} we find:
\begin{align}\label{eq:Z_calFH}
     \mathcal{Z}_\vpar &= \frac{1}{4b^2}\left[\left(\frac{\partial_\beta \mathcal{F}}{\mathcal{F}}\right)^2  - 2   
    \frac{\partial^2_\beta \mathcal{F}}{\mathcal{F}} 
    \right]
    +\frac{1-\vpar^2}{4}\left[ \left(\frac{\partial_\ell \mathcal{F}}{\mathcal{F}}\right)^2 - 2  \frac{\partial^2_\ell \mathcal{F}}{\mathcal{F}}  
    \right] \nonumber\\
    &\qquad +
    \frac{\vpar^2}{4}\left[2\frac{H''}{H}-\left(\frac{H'}{H}\right)^2
    +\frac{H'}{H}\frac{\partial_\ell \mathcal{F}}{\mathcal{F}}\right]
    \,,
\end{align}
from which the limiting cases $\vartheta =0,1$ can easily be read off.

\begin{lemma}\label{lem:Z}
Consider a Lorentzian metric of the form \eqref{eq:region2General}, such that
\begin{enumerate}
  
    \item\label{hyp:smooth} $H$ and $\mathcal{F}$ are smooth functions (at least $C^2$) in $\ell\in(\Lint,\Lext)$ 
    
    \item\label{hyp:interp} and they satisfy the following conditions\footnote{
        Notice that these belong to the set of conditions \eqref{eq:calFH_interpcond} needed to have a continuous matching between the regions $\Omega_2$ and $\Omega_1$.
        }
    at $\ell = \Lext$
    \begin{align}
        \mathcal{F}(\Lext, \beta) &= (\Lext + b \cos(\beta))^2\qquad \forall \beta\in[0,2\pi)\,,\label{eq:limRext_calFx}\\
        H(\Lext) &= 1\,.
    \end{align}
    \end{enumerate}
    Then the following asymptotic behavior holds:
    \begin{align}
        \mathcal{Z}_\vpar &\sim   
        - \frac{1}{b(\Lext-b)} +\frac{1-\vpar^2}{4}\left[\left(\frac{\partial_\ell \mathcal{F}(\Lext,\pi)}{(\Lext-b)^2}\right)^2 -2 \frac{\partial^2_\ell \mathcal{F}(\Lext,\pi)}{(\Lext-b)^2}\right] \nonumber\\
        &\qquad +
        \frac{\vpar^2}{4}\left[2H''(\Lext)- H'(\Lext)^2
        +\frac{H'(\Lext) \partial_\ell \mathcal{F}(\Lext,\pi)}{(\Lext-b)^2}\right]
        \qquad  (\ell\to \Lext, \beta\to \pi)\,,\label{eq:lemma_Z}
    \end{align}
    where the suppressed terms are order 1 in $\beta-\pi$ and/or order 1 in $\ell-\Lext$.
\end{lemma}
\begin{proof}
    From \eqref{eq:Z_calFH}, and using that the involved functions are at least $C^2$ (Hypothesis~\ref{hyp:smooth}) one can directly obtain the leading order behavior:
    \begin{align}
        \mathcal{Z}_\vpar &\sim   
        \frac{1}{4b^2}\left[\left(\frac{\partial_\beta \mathcal{F}(\Lext,\pi)}{(\Lext-b)^2}\right)^2 
         - 2\frac{\partial^2_\beta \mathcal{F}(\Lext,\pi)}{ (\Lext-b)^2}\right] \nonumber \\
         &\qquad +\frac{1-\vpar^2}{4}\left[\left(\frac{\partial_\ell \mathcal{F}(\Lext,\pi)}{(\Lext-b)^2}\right)^2 -2 \frac{\partial^2_\ell \mathcal{F}(\Lext,\pi)}{(\Lext-b)^2}\right] \nonumber\\
        &\qquad +
        \frac{\vpar^2}{4}\left[2H''(\Lext)- H'(\Lext)^2
        +\frac{H'(\Lext) \partial_\ell \mathcal{F}(\Lext,\pi)}{(\Lext-b)^2}\right]
        \qquad  (\ell\to \Lext, \beta\to \pi)\,,
    \end{align}
    where we recall that we are abbreviating $\partial^k_i \mathcal{F}(\Lext,\pi) := \lim_{\ell \nearrow \Lext}\lim_{\beta \to \pi} \partial^k_i \mathcal{F}(\ell,\beta)$ and the suppressed terms are order 1 in $\beta-\pi$ and/or order 1 in $\ell-\Lext$. To derive this expression we just made use of the interpolation conditions in Hypothesis~\ref{hyp:interp}.
    
    Furthermore, the derivatives with respect to $\beta$ should tend to the derivatives of the limiting function \eqref{eq:limRext_calFx} as a consequence of Hypothesis~\ref{hyp:smooth}, namely
     \begin{equation}
    \partial^k_\beta \mathcal{F}(\Lext,\pi) =  \lim_{\beta \to \pi} \partial^k_\beta \left(\lim_{\ell \nearrow \Lext} \mathcal{F}(\ell, \beta)\right) 
    = \lim_{\beta \to \pi} \partial^k_\beta (\Lext + b \cos(\beta))^2\,,
    \end{equation}
    at least for $k=1,2$. The results are:    
    \begin{equation}
        \partial_\beta \mathcal{F}(\Lext,\pi) = 0\,,\qquad
        \partial^2_\beta \mathcal{F}(\Lext,\pi)  =2b(\Lext - b)\,.
    \end{equation}
    After substituting this into the previous asymptotic expansion for $\mathcal{Z}_\vpar$ we reach the expression that we wanted to prove.
\end{proof}

If we impose an additional condition on the partial derivatives with respect to $\ell$ we can find the following stronger result:
\begin{thm}\label{thm:WECNECviol}
Consider any Lorentzian metric of the form \eqref{eq:region2General} satisfying the hypotheses of Lemma~\ref{lem:Z} together with
\begin{enumerate}
    
    \item the following condition
    as $\ell \to \Lext$
    \begin{equation}\label{eq:limRext_dlcalF}
        \partial_\ell\mathcal{F}(\Lext, \beta) = 
\partial^2_\ell \mathcal{F}(\Lext, \beta)=0\qquad\forall \beta\in[0,2\pi)
    \end{equation}

    \item and $\Lext> b$ (non-degenerate torus).
    
\end{enumerate}
Then, in some open set sufficiently close to $\ell=\Lext$ and $\beta=\pi$:
\begin{enumerate}
    \item[(i)] the WEC is violated;
    \item[(ii)] additionally, if 
    \begin{equation}\label{eq:limRext_H}
        H'(\Lext) = H''(\Lext) = 0\,,
    \end{equation}
    the NEC is violated.
\end{enumerate}
\end{thm}

\begin{proof}
Since we are under the hypothesis of Lemma~\ref{lem:Z}, we can make use of the asymptotic expansion \eqref{eq:lemma_Z}. After using Eq.~\eqref{eq:limRext_dlcalF}, we find
\begin{equation}\label{eq:asymp_FG}
    \mathcal{Z}_\vpar \sim   - \frac{1}{b(\Lext-b)} + \frac{\vpar^2}{4}(2H''(\Lext)-H'(\Lext)^2) \qquad  (\ell\to \Lext, \beta\to \pi)\,.
\end{equation}
In the case $\vpar=0$ we get
\begin{equation}\label{eq:auxthm}
    \mathcal{Z}_0 \sim   - \frac{1}{b(\Lext-b)}\qquad  (\ell\to \Lext, \beta\to \pi)\,,
\end{equation}
which is negative since we are assuming $\Lext>b$. This implies that the WEC is violated, as expected based on the intuition developed with the infinite torus example. Finally, in the case $\vpar=1$ together with \eqref{eq:limRext_H} we get exactly the same leading order term appearing in Eq.~\eqref{eq:auxthm}, so the NEC is violated in the asymptotic region.
\end{proof}

By adapting this result to the particular case \eqref{eq:FG} we can derive the corollary:
\begin{cor} \label{cor:WECNECviol}
    For any Lorentzian metric of the form \eqref{eq:gregion2} with 
\begin{enumerate}

    \item $\Lext> b$ (non-degenerate torus),
    
    \item $H$, $F$ and $G$ smooth functions (at least $C^2$) in $\ell\in(\Lint,\Lext),$
    
    \item $G(\Lext)=0$ and $F(\Lext)=H(\Lext)=1$,

    \item and $F'(\Lext)= F'' (\Lext)= G'(\Lext)= G''(\Lext)=0$,
\end{enumerate}
in some open set sufficiently close to $\ell=\Lext$ and $\beta=\pi$:
\begin{enumerate}
    \item[(i)] the WEC is violated;
    \item[(ii)] additionally, if $H'(\Lext) = H''(\Lext) = 0$, the NEC is violated.
\end{enumerate}
\end{cor}

\subsection{Some specific models}
\label{Subsec:specific}

For simplicity, we now focus on the lowest-degree polynomial ansatz of the type \eqref{eq:FG} that satisfies the interpolation conditions  \eqref{eq:FGH_interpcond}, the no-internal-shell condition \eqref{eq:no-shell} and the hypothesis of Theorem~\ref{thm:WECNECviol} (equivalently, those of the Corollary~\ref{cor:WECNECviol}). This can be obtained by Hermite interpolation and the result reads:
 \begin{align}
     F(\ell) &= \frac{(\ell-\Lint)^2\big[3\ell^2-2\ell(4\Lext-\Lint)+6\Lext^2-4\Lext\Lint+\Lint^2\big]}{(\Lext-\Lint)^4}\,,\nonumber\\
     G(\ell) &= \dfrac{m^2}{\pi^2} \frac{(\Lext-\ell)^3(3\ell+\Lext-4\Lint)}{(\Lext-\Lint)^4}\,,\nonumber\\
     H(\ell)&=\frac{2\ell(\Lext-\ell)^3+2\ell\Lint(\ell^2-3\ell\Lext+3\Lext^2)-\Lint(\Lext^3+3\Lext^2\Lint-3\Lext\Lint^2+\Lint^3)}{\Lint(\Lext-\Lint)^3}\,,\label{eq:greatFGH}
\end{align}
which are strictly positive for $\ell\in(\Lint,\Lext)$.

In the following, we analyze violations of the energy conditions throughout the entire $\Omega_2$ region for this specific choice. In Fig.~\ref{fig:Z0FH0}, we represent the function $\mathcal{Z}_0$ (essentially, the energy density seen by observers in the orbits of the vector in Eq.~\eqref{eq:observertau}) as a function of the coordinates $\ell$ and $\beta$ of the interpolating region  for a particular set of parameters. The contour plot clearly shows that the negative energy region enclosed by the surface of vanishing energy density (the dashed line) will be reached by any observer sufficiently close to $\ell = \Lext$ and $\beta = \pi$. 

We also present, in Fig.~\ref{fig:Z1FH0}, the scalar $\mathcal{Z}_1$ as a function of $\ell$ and $\beta$ for the same particular set of parameters. The dark regions represent those in which the NEC is violated.

\begin{figure}[H]
    \centering
    \includegraphics[scale=0.48]{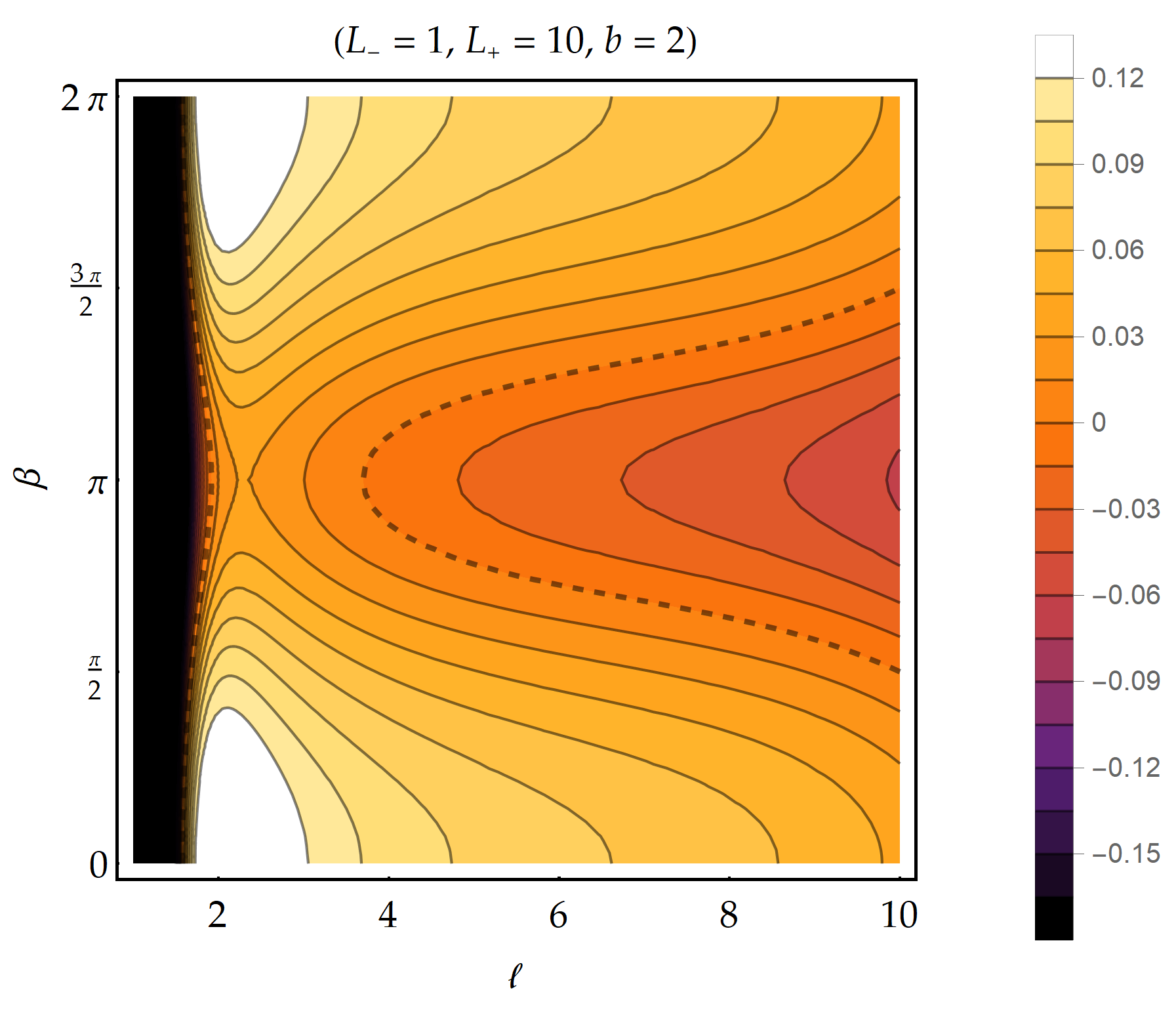}
    \includegraphics[scale=0.48]{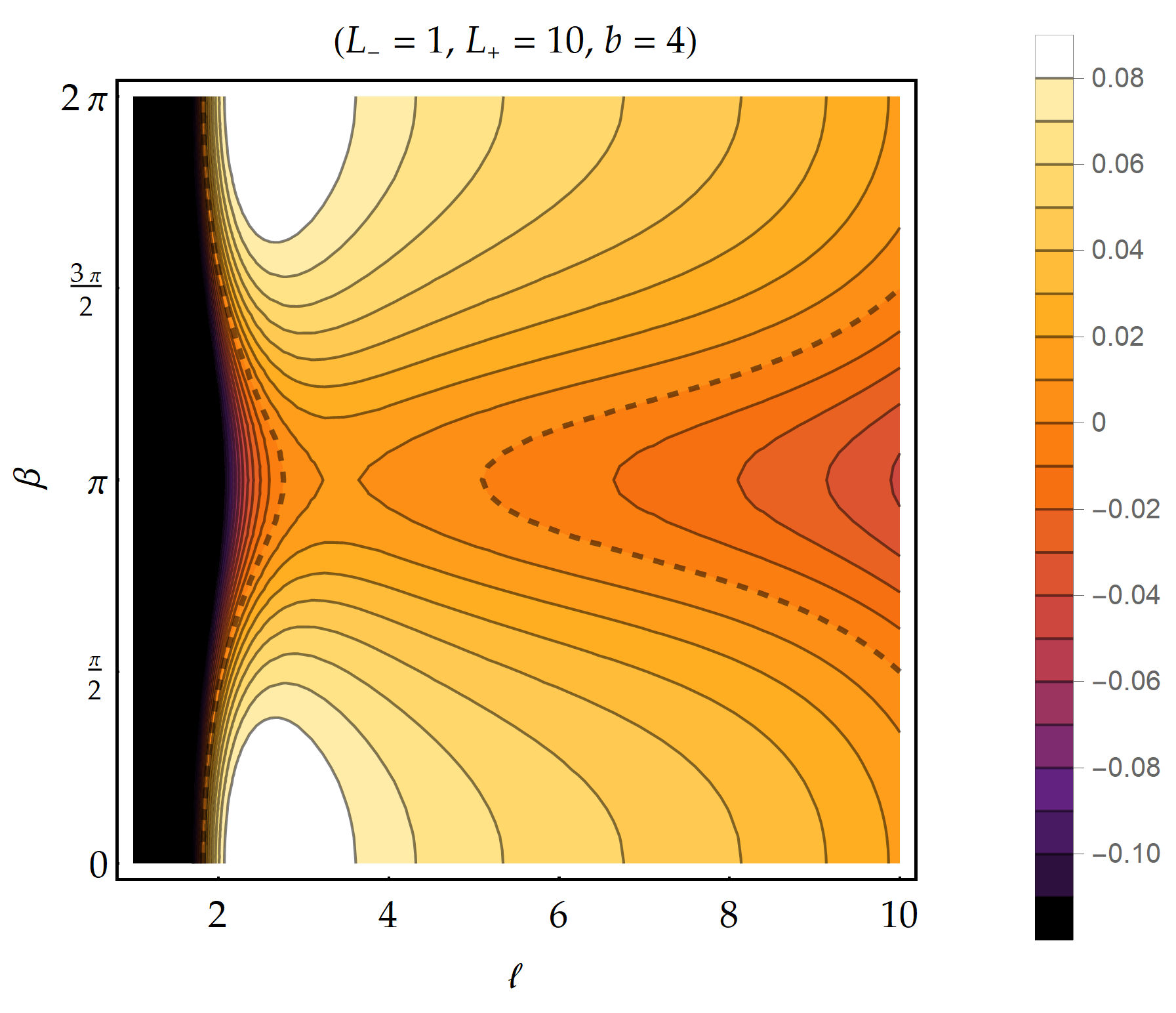}
    \includegraphics[scale=0.48]{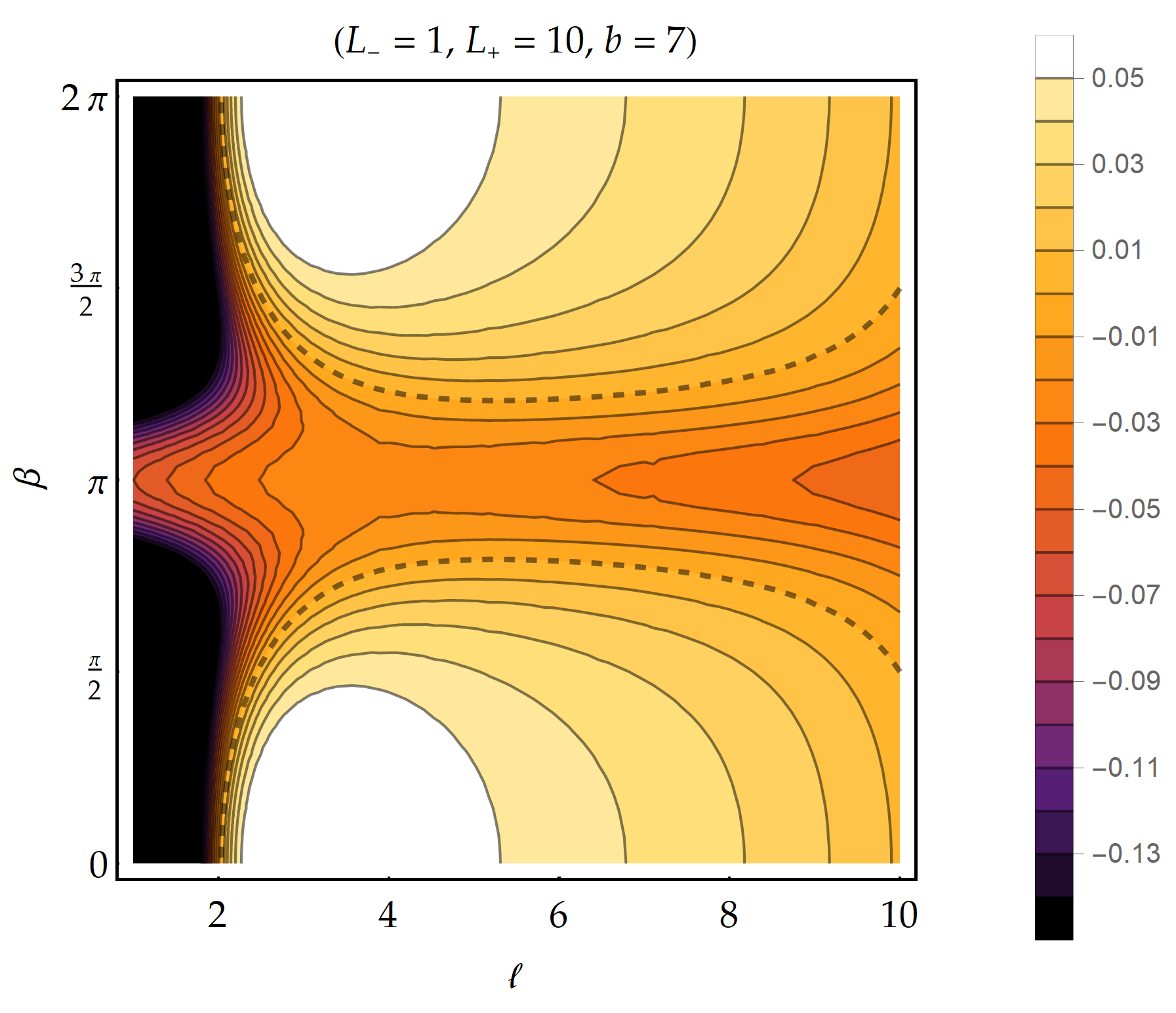}
    \includegraphics[scale=0.48]{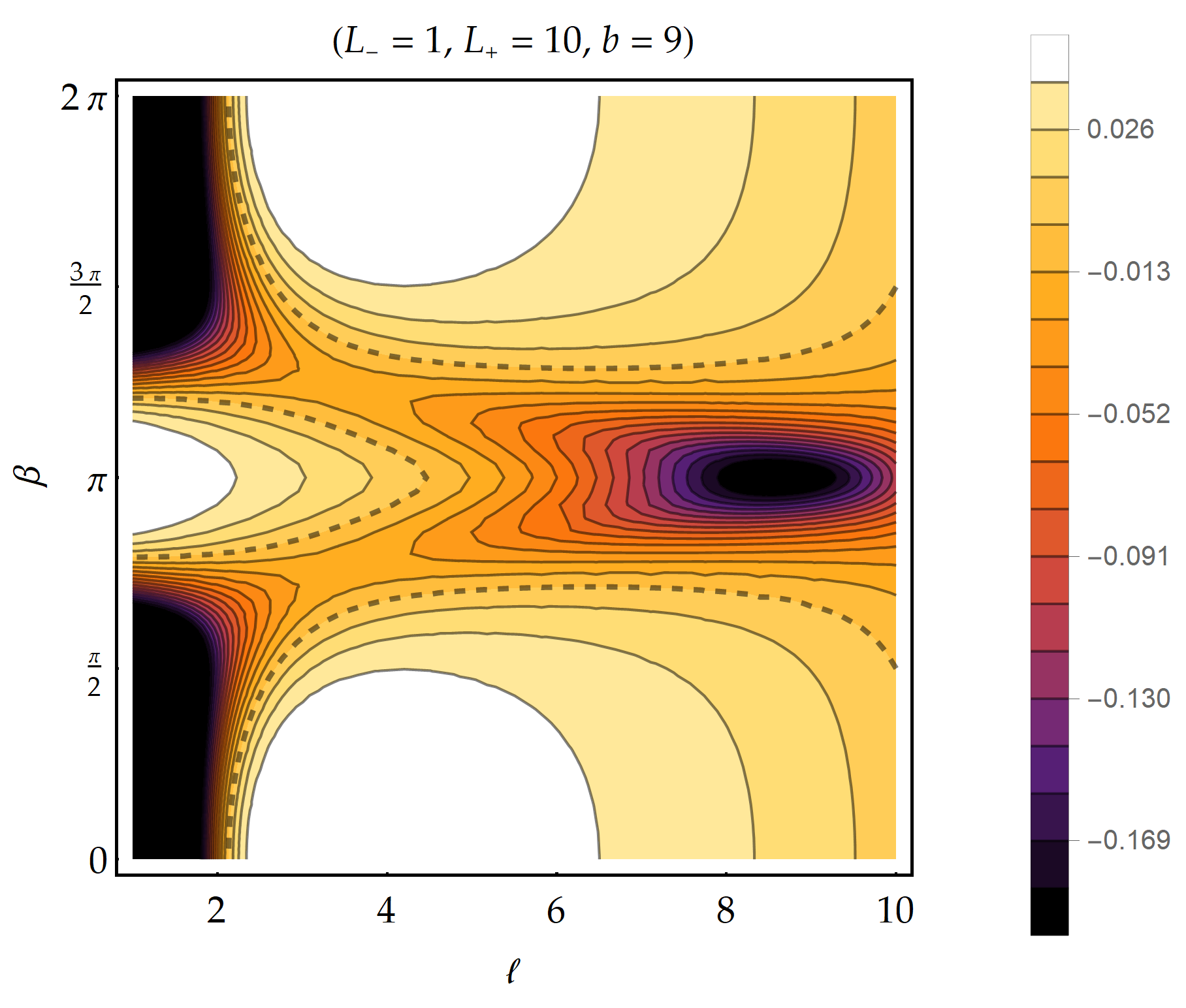}
    \caption{Contour plot for the values of $\mathcal{Z}_0(\beta, \ell)$ within the interpolating region $\Omega_2$ for the ansatz \eqref{eq:greatFGH} and for different values of $b$. The dashed line represents the limit between positive and negative energy densities ($\mathcal{Z}_0=0$). For the plots, we took $\Lint=1$ $\Lext=10$.}
    \label{fig:Z0FH0}
\end{figure}

\begin{figure}[H]
    \centering
    \includegraphics[scale=0.48]{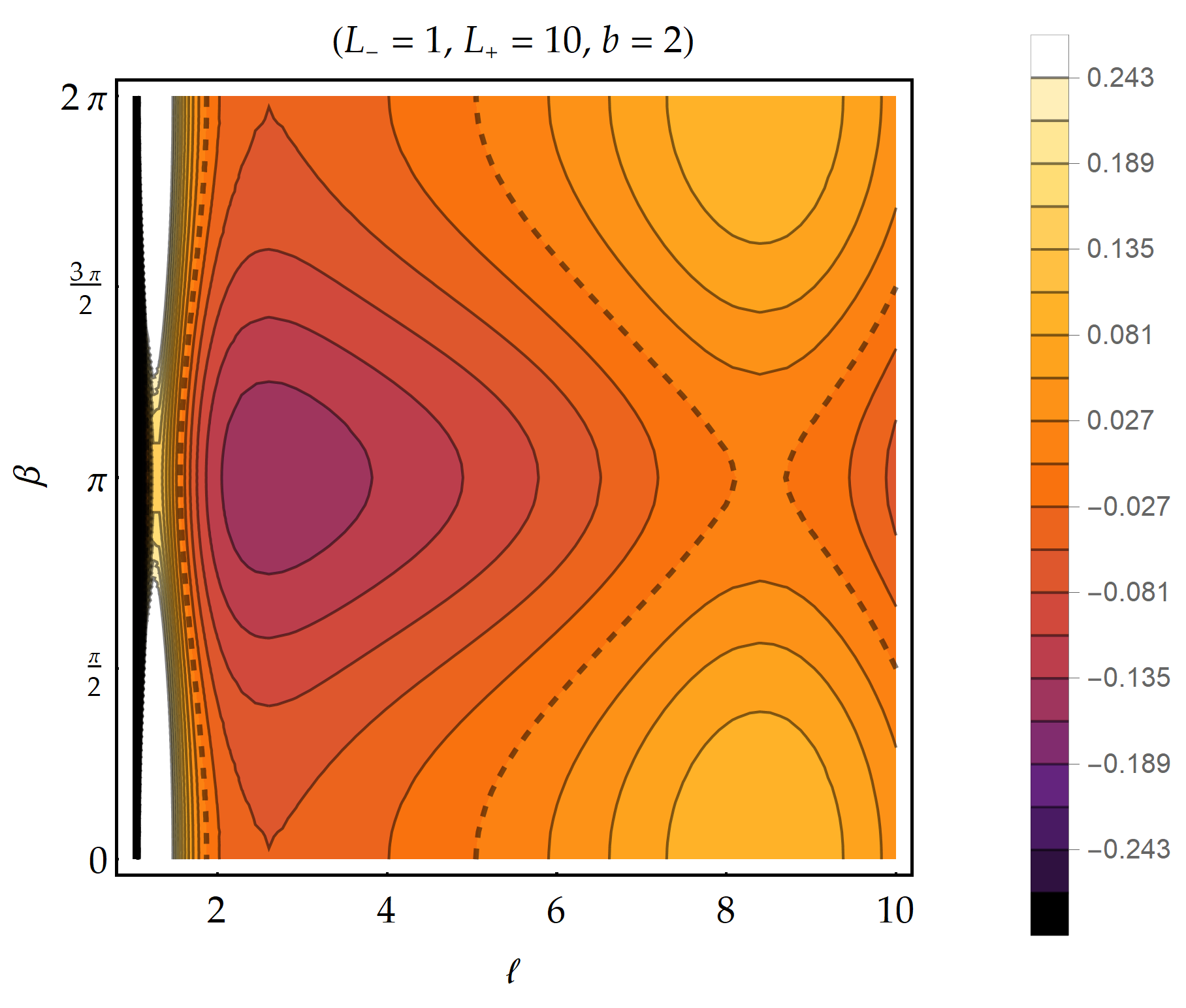}
    \includegraphics[scale=0.48]{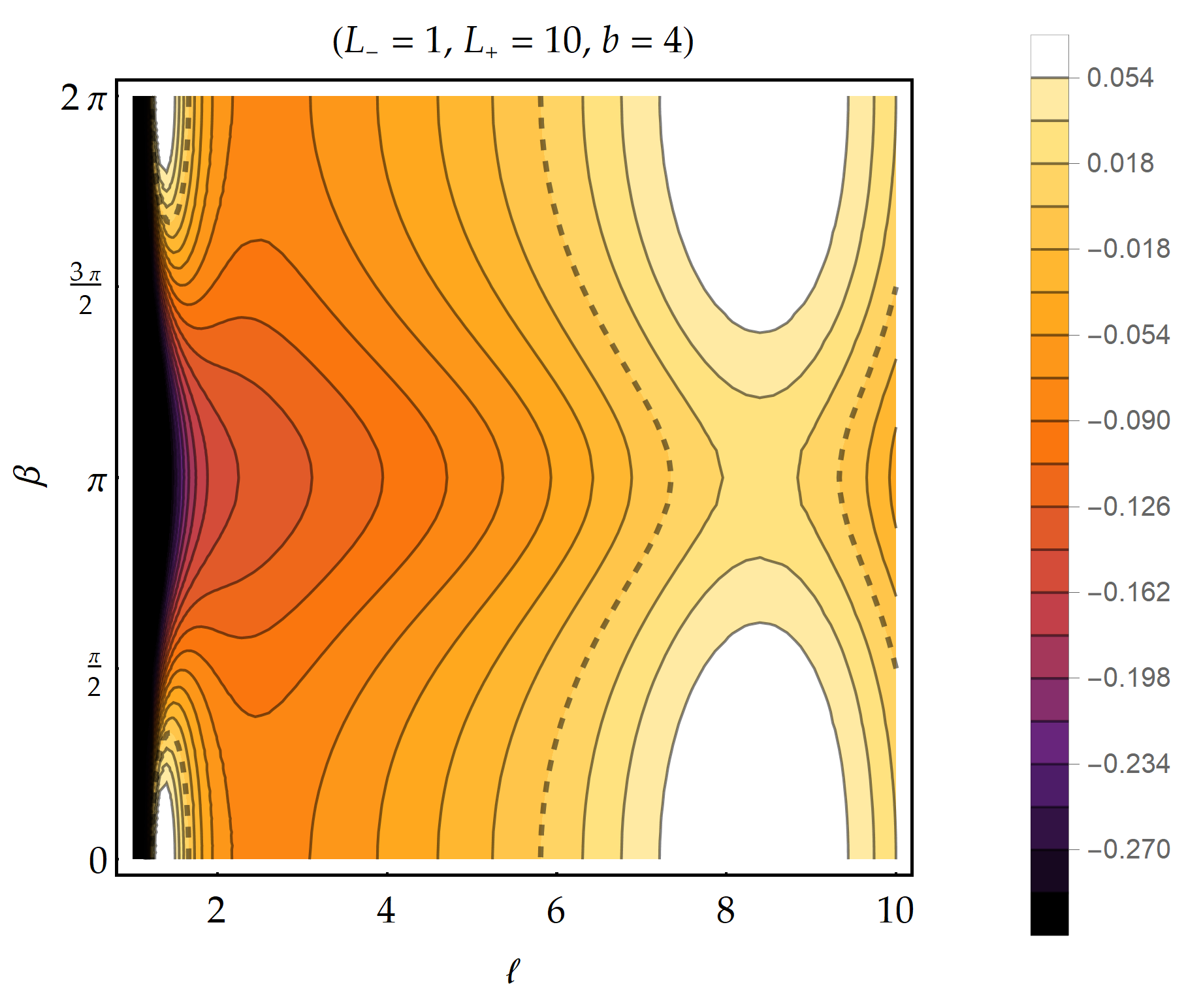}
    \includegraphics[scale=0.48]{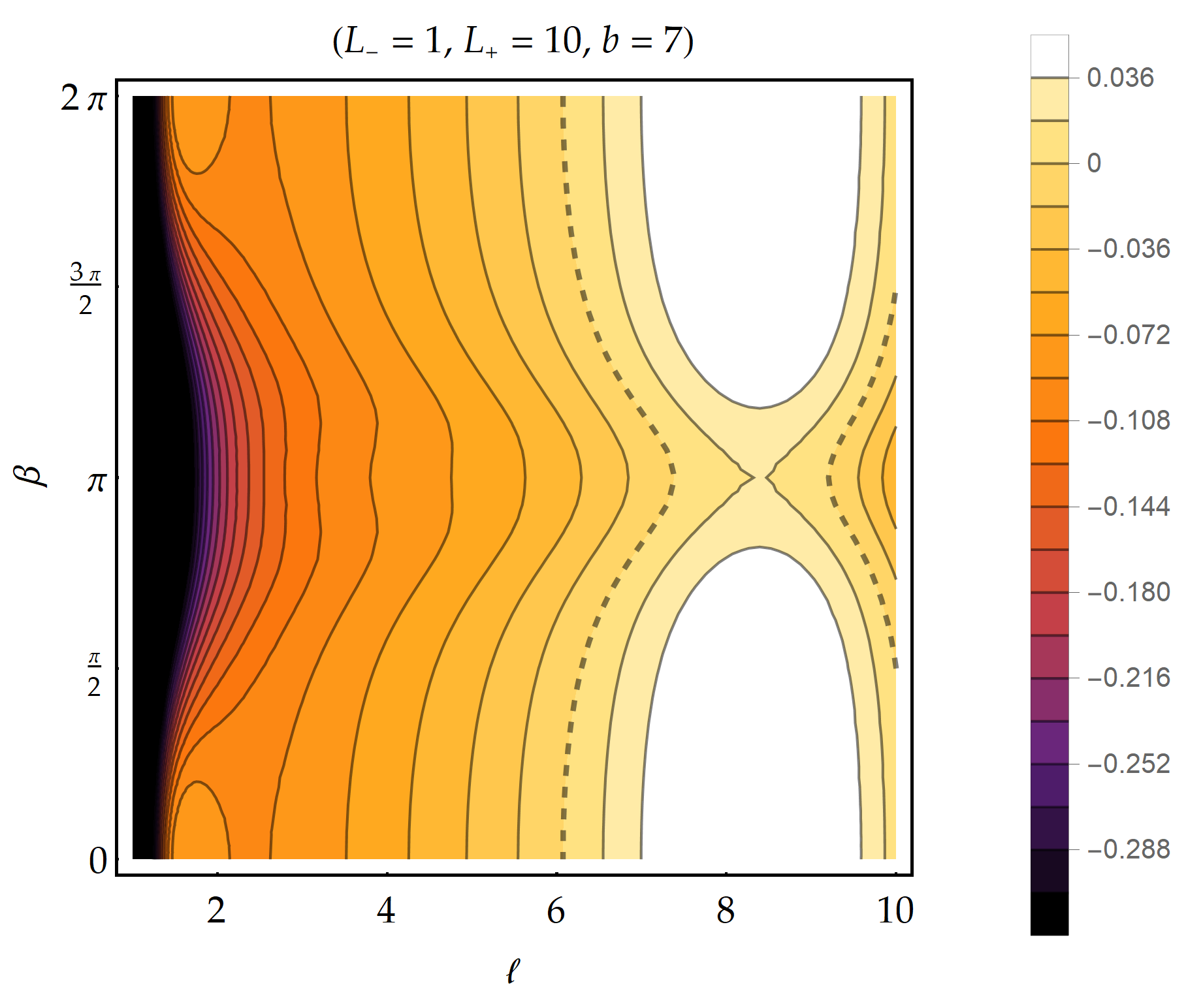}
    \includegraphics[scale=0.48]{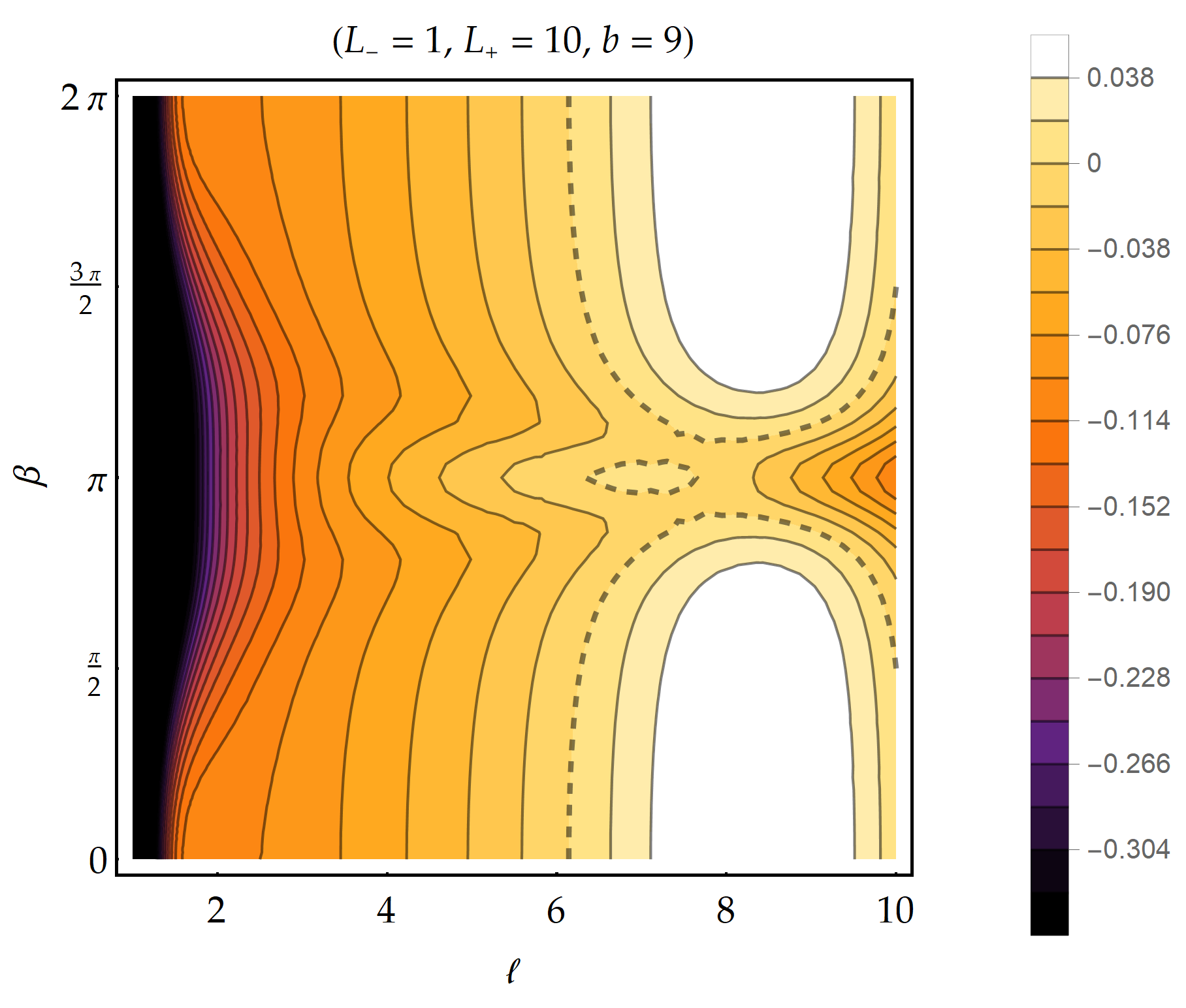}
    \caption{Contour plot for the values of $\mathcal{Z}_1(\beta, \ell)$ within the interpolating region $\Omega_2$ for the ansatz \eqref{eq:greatFGH} and for different values of $b$. The dashed line represents the limit between positive and negative values ($\mathcal{Z}_1=0$). For the plots, we took $\Lint=1$ $\Lext=10$.}
    \label{fig:Z1FH0}
\end{figure}

\section{Conclusions}
\label{Sec:Conclusions}

In this work we have provided a comprehensive overview of toroidal black holes in 4D. We have revisited the assumptions  behind the classic theorems that forbid their existence with a specific emphasis on the energy condition violations. We have also reviewed previous attempts to build toroidal black holes in 4D and how and why such attempts failed in providing a regular external region. 

We have presented the simplest construction of a toroidal black hole, achieved by matching three regions: an external flat spacetime, an inner region described by the simplest locally vacuum toroidal black hole (a Rindler line element with periodically identified transverse directions), and an intermediate region containing a nontrivial matter profile that bridges the two. We analyzed the properties of the matter content in the intermediate region, which consists of an anisotropic fluid and two shells located at the boundaries where the intermediate region joins the inner and outer regions. While a specific choice of the interpolating matter profile can eliminate the internal shell, the external shell remains unavoidable within the family of geometries that we consider. Although it might be possible to consider alternative interpolations for which it is possible to avoid the external thin-shell, we have decided to stick to the particular one considered in this paper for its easier interpretation. Additionally, we explained why a simpler construction directly matching the inner toroidal black hole to the external Minkowski spacetime via a single shell fails: the shell's induced metric would be curved on one side and flat on the other, violating the first of Darmois-Israel junction conditions, which requires the continuity of the first fundamental form across the shell.

We have explicitly studied the violations of energy conditions and confirmed that the DEC is violated, in agreement with the implications of Hawking's theorem. In particular, we have proved the violation of the WEC and NEC near and on the external shell as approached from the interpolating region and inspected the properties of the fluid in the intermediate region for some specific choices of the interpolating functions.

There are several promising directions for future research building on this work. First, the assumption of a flat external region is a significant simplification, and it would be interesting to replace it with a spacetime that is only asymptotically flat, such as the Schwarzschild or Curzon solution. In particular, it would be interesting to analyze the line element used in~\cite{Kleihaus2019,Cunha2024} to understand whether it can represent toroidal horizons that are locally in vacuum and whether our models could be brought to that form. Additionally, the analysis could be extended to the more general case of stationary and axisymmetric configurations. Given the integrability of the Ernst equation, addressing this problem in such a setup seems feasible. However, it would require a comprehensive characterization of local stationary and axisymmetric vacuum black holes, analogous to the Geroch-Hartle framework for the static case. Finally, the results discussed here are restricted to asymptotically flat scenarios, and it remains an open question whether energy condition violations are still necessary to support these configurations in alternative contexts, such as black holes embedded in cosmological spacetimes.

\acknowledgments{
The authors would like to thank Luis J. Garay, Matt Visser and Bert Janssen for valuable discussion and Jos\'e M M Senovilla for helpful correspondence. Financial support was provided by the Spanish Government through the Grants No. PID2020-118159GB-C43, PID2020-118159GB-C44, PID2023-149018NB-C43 and PID2023-149018NB-C44 (funded by \mbox{MCIN/AEI/10.13039/501100011033),} and by the Junta de Andaluc\'ia through the project FQM219. GGM is funded by the Spanish Government fellowship FPU20/01684. CB and GGM acknowledge financial support from the Severo Ochoa grant CEX2021-001131-S funded by MCIN/AEI/10.13039/501100011033. The work of AJC has been supported by the grant PID2022-140831NB-I00 funded by MICIU/AEI/10.13039/ 501100011033. The computations have been checked with xAct \cite{xAct}, a Mathematica package for tensorial symbolic calculus. The corresponding notebook is available upon request.
}

\bibliographystyle{JHEP}
\bibliography{torus_biblio}

\end{document}